\documentclass{article}

\pdfoutput=1
\usepackage{amsmath}
\usepackage{amsthm}
\usepackage{amsfonts}
\usepackage{graphicx}
\usepackage{caption}
\usepackage{subcaption}
\usepackage{setspace}
\usepackage{epstopdf}
\usepackage{cleveref} 
\usepackage{wrapfig}
\usepackage{float}
\usepackage{fullpage}

\theoremstyle{definition}
\newtheorem{theorem}{Theorem}[section]
\newtheorem{remark}{Remark}[section]
\newtheorem{proposition}{Proposition}[section]
\newtheorem{definition}{Definition}[section]
\newtheorem{lemma}{Lemma}[section]
\newtheorem{corollary}{Corollary}[section]
\newtheorem{example}{Example}[section]

\newtheorem{algorithm}{Algorithm}[section]

\newcommand{\mm}{{\mathcal{M}}}

\newcommand{\per}{\mathrm{Per}}

\newcommand{\mf}{\mathcal{F}}

\newcommand{\So}{\mathbb{S}^1}

\begin{document}

\title{The excluded volume of two-dimensional convex bodies: Shape reconstruction and non-uniqueness}
\author{Jamie M. Taylor\footnote{Address for correspondance :Department of Mathematical Sciences/Liquid Crystal Institute, Kent State University, Kent, Ohio 44242, USA. Email: {\tt jtayl139@kent.edu}}}\date{}\maketitle
\begin{abstract}
In the Onsager model of one-component hard-particle systems, the entire phase behaviour is dictated by a function of relative orientation, which represents the amount of space excluded to one particle by another at this relative orientation. We term this function the excluded volume function. Within the context of two-dimensional convex bodies, we investigate this excluded volume function for one-component systems addressing two related questions. Firstly, given a body can we find the excluded volume function?, Secondly, can we reconstruct a body from its excluded volume function? The former is readily answered via an explicit Fourier series representation, in terms of the support function. However we show the latter question is ill-posed in the sense that solutions are not unique for a large class of bodies. This degeneracy is well characterised however, with two bodies admitting the same excluded volume function if and only if the Fourier coefficients of their support functions differ only in phase. Despite the non-uniqueness issue, we then propose and analyse a method for reconstructing a convex body given its excluded volume function, by means of a discretisation procedure where convex bodies are approximated by zonotopes with a fixed number of sides. It is shown that the algorithm will always asymptotically produce a best $L^2$ approximation of the trial function, within the space of excluded volume functions of centrally symmetric bodies. In particular, if a solution exists, it can be found. Results from a numerical implementation are presented, showing that with only desktop computing power, good approximations to solutions can be readily found.
\end{abstract}
KW: Excluded volume; Hard-particle systems; Onsager model; Virial Expansion

\section{Introduction}

\subsection{Motivation}

In many particle systems, particles are viewed to interact through long-range, attractive interaction, and short-range, repulsive interaction \protect\cite{gay1981modification,jones1924determination}. The short-range repulsive effects are often idealised in the hard particle limit, whereby particles only interact insofar as they are not permitted to overlap, but do not interact otherwise \protect\cite{hansen1986ir}. These kinds of interactions are geometric in nature, with particle shape driving the phase behaviour of such systems. It is now well known that haad particle systems of aspherical bodies such as ellipsoids \protect\cite{camp1996hard,frenkel1985hard}, spherocylinders \protect\cite{somoza1990nematic,veerman1990phase} and more exotic shapes \protect\cite{damasceno2012predictive} can form liquid crystalline mesophases. One of the earliest, and most successful, attempts to capture the interplay between particle shape and phase behaviour was by Onsager \protect\cite{onsager1949effects}, describing the phase transition of elongated rods from a disordered (isotropic) phase, to an ordered (nematic) phase with increasing concentration. The mean-field nature of the model simplifies the complex behaviours into a far more elegant and tractable problem. The Onsager model is intimately related to the virial expansion of hard particles, and describes the leading order behaviour of dilute systems \protect\cite{masters2008virial}. The model's key tool is to quantify the notion of the amount of volume made inaccessible to a probe particle by a typical particle in the system, the so-called excluded volume. For a one-component system, this information can be encoded via a function of the relative orientation of two particles, outputting the volume excluded to one particle by the other. Symbolically, if $\mathrm{SO}(n)$ is the space of proper rotations in $\mathbb{R}^n$, for a rigid body idealised as a subset of $\mathbb{R}^n$, we may define a function $Ex(\cdot,\mm):\mathrm{SO}(n)\to\mathbb{R}$ so that $Ex(R,\mm)$ denotes the volume inaccessible to translates of $R\mm$ given the presence of $\mm$. As this quantity is of fundamental importance to hard-particle systems, much work has been devoted to calculating and analysing excluded volumes of shapes at given orientations, through studies on axially symmetric convex bodies \protect\cite{piastra2013octupolar,piastra2015explicit}, spherozonotopes \protect\cite{mulder1986solution,mulder2005excluded}, and studies into more general and complex shapes, including non-convex bodies \protect\cite{bisi2012excluded,palffy2014minimum,xu2014microscopic}. Convex rigid bodies provide a vastly simpler testing ground for the hard-particle theory, and for a contemporary review on existing results see \protect\cite{singh2001molecular} and references therein. 

 Within this work, we aim to address an inverse-type problem related to the excluded volume, asking if it is possible to recover particle shape knowing only the function which outputs the excluded volume at a given orientation. This question is inspired by recent interest in the design of materials.

\subsection{Preliminaries and notation}

\subsubsection{Excluded volumes and Onsager}

Given two arbitrary bodies, $\mm_1,\mm_2$, represented as compact subsets of $\mathbb{R}^n$, their excluded volume is given by ${\mathcal{V}(\mm_1,\mm_2)=|\mm_1-\mm_2|}$ \protect\cite{mulder2005excluded}. This should be interpreted as the volume made unavailable to translates of the body $\mm_1$, due to the presence of the body $\mm_2$. When the two bodies are rotations of an identical particle $\mm$, we define the excluded volume function $Ex(\cdot,\mm):\mbox{SO}(n)\to\mathbb{R}$ by $Ex(R,\mm)=|\mm-R\mm|$. Here $R\in\mbox{SO}(n)$ should be interpreted as their relative orientation. The Onsager functional aims to describe phase transitions in particle systems by means of purely steric, repulsive interactions. For a spatially homogeneous system of rigid particles, described by a single particle orientation distribution function $\rho:\mbox{SO}(n)\to\mathbb{R}$, the Onsager free energy functional is given by \protect\cite{onsager1949effects}
\begin{equation}\begin{split}
&\mf(\rho,\mm)\\
=&\int_{\mbox{SO}(n)}\rho(R)\ln \rho(R)\,dR 
+ \frac{1}{2}\int_{\mbox{SO}(n)}\int_{\mbox{SO}(n)}Ex(S^TR,\mm)\rho(R)\rho(S)\,dR\,dS.
\end{split}\end{equation}
The density constraint $\int_{\mbox{SO}(n)}\rho(R)\,dR=\rho_0$ is imposed. The left-hand Shannon entropy term of the energy favours a disordered configuration, $\frac{\rho_0}{|\mbox{SO}(n)|}=\rho(R)$, while the right-hand term generally favours alignment. In the prototypical case of convex, axially and head-to-tail symmetric bodies in three dimensions the right hand term favours parallel alignment of particles \protect\cite{palffy2014minimum}. The competition is mediated by concentration, favouring ordered configurations in dense regimes. Even for particles in three dimensions lacking rotational symmetry, the excluded volume function can prove elusive, so for much of our analysis we constrain ourselves to only consider convex bodies in two dimensions, in which case $\mbox{SO}(2)\approx \mathbb{S}^1$, and we represent $R\in\mbox{SO}(2)$ simply by a scalar angle $\theta$. 

There are three significant reasons why a two dimensional system is a mathematically simpler scenario. Firstly, $\mbox{SO}(2)$ is a one-dimensional manifold, and can easily be represented by a single non-degenerate chart as $[0,2\pi]$ (the angle of rotation) for many purposes. Secondly, $\mbox{SO}(2)$ is an abelian group, which does not hold in higher dimensions, simplifying the algebraic aspects of the problem. Thirdly, area is a quadratic-type quantity in two dimensions, making Fourier approaches particularly elegant and appropriate. Despite these many simplifying properties, the two-dimensional testing ground provides a rich behaviour that will hopefully motivate future work in higher dimensions. The validity of Onsager in two-dimensions is suspect due to the Mermin-Wagner theorem, stating that long-range order is unstable against low frequency fluctuations in systems with the symmetries we consider {\protect\cite{mermin1966absence}, however we expect the excluded area function to be relevant in the local ordering of hard particle systems nonetheless, of mathematical interest, and an important stepping stone in developing techniques applicable higher dimensional scenarios.

\subsubsection{Convex geometry preliminaries}
In this introduction we freely quote results from \protect\cite{schneider2014convex}.

Let $\mathcal{K}^n$ denote the convex compact bodies in $\mathbb{R}^n$. Given a body $\mm\in\mathcal{K}^n$, we define its support function $h(\cdot,\mm):\mathbb{S}^{n-1}\to\mathbb{R}$ by 
\begin{equation}\begin{split}
h(u,\mm)=\max\limits_{x \in \mm} u \cdot \mm.
\end{split}\end{equation}
Geometrically, $h(u,\mm)$ denotes the distance from the origin to a tangent plane of $\mm$ orthogonal to $u$, and provides a parameterisation of a convex body. The definition readily extends to $u\in\mathbb{R}^n$, but this merely defines a 1-homogeneous extension of $h$. In 2D, we will often abuse notation and denote the argument of $h(\cdot,\mm)$ simply by a scalar angle $\theta$, so $h(\theta,\mm)=h(u_\theta,\mm)$ with $u_\theta=\cos\theta e_1+\sin\theta e_2$. The support function of a convex body is the Legendre transform of the indicator function of the body $\chi_{\mm}$, given by $\chi_{\mm}(x)=0$ if $x \in \mm$ and $\chi_{\mm}(x)=+\infty$ otherwise. A closed convex body can be reconstructed from the support function as an intersection of half-spaces as 
\begin{equation}\begin{split}
\mm = \bigcap\limits_{u \in \mathbb{S}^{n-1}} \left\{ x \in \mathbb{R}^n : x\cdot u \leq h(u,\mm)\right\}.
\end{split}\end{equation}

Given two bodies $\mm_1,\mm_2$, their Minkowski sum is defined as ${\mm_1+\mm_2=\{x+y:x\in\mm_1,y\in\mm_2\}}$. A Minkowski sum of convex bodies is itself convex. An advantage of using the support function to describe convex bodies that it respects Minkowski addition, in the sense that if $\mm_1,\mm_2$ are convex bodies, then
\begin{equation}\begin{split}
h(u,{\mm}_1+{\mm}_2)=h(u,{\mm}_1)+h(u,{\mm}_2).
\end{split}\end{equation}
Furthermore support functions respect rotations in the sense that for any $R \in \mbox{O}(n)$, 
\begin{equation}\begin{split}
h(u,R\mm)=h(R^Tu,\mm).
\end{split}\end{equation}

We will need a notion of convergence of convex bodies, for which the Hausdorff metric is ideally suited. The Hausdorff metric is defined as 
\begin{equation}\begin{split} 
&d_H({\mm}_1,{\mm}_2)
=\max\left(\max\limits_{x\in \mm_1}\min\limits_{y \in \mm_2} |x-y|,\max\limits_{y\in \mm_2}\min\limits_{x \in \mm_1} |x-y| \right)
\end{split}\end{equation}
Equivalently, $d_H({\mm}_1,{\mm}_2)<\epsilon$ if and only if $\mm_1\subset \mm_2+\epsilon B$ and $\mm_2\subset \mm_1 + \epsilon B$ where $B$ is the ball of radius $1$ and $\epsilon>)$. The space $\mathcal{K}^n$ is closed with respect to the Hausdorff metric, and compact when restricted to uniformly bounded subsets. Furthermore, the Hausdorff distance can be calculated from the support function as $||h(\cdot,\mm_1)-h(\cdot,\mm_2)||_\infty = d_H(\mm_1,\mm_2)$.

 For the sake of this work, unless stated otherwise we will take all convex bodies to be centred at the origin for simplicity, and without loss of generality. 
 
\subsubsection{$C^2_+$ bodies}

A particular class of useful convex bodies will be the $C^2_+$ bodies, defined to be the strictly convex bounded sets that have a $C^2$ boundary and non-vanishing curvature. These bodies are important because they are in many senses well behaved, and significantly they are dense in $\mathcal{K}^n$ with respect to Hausdorff metric. Equivalently, their support functions are dense in the set of all possible support functions, with respect to $L^\infty$ convergence. In $\mathbb{R}^2$, the main focus of this work, $C^2_+$ bodies are characterised in a simple way by their support functions, in that $\mm\in C^2_+$ if and only if $h(\cdot,\mm)$ is $C^2$ and satisfies the differential inequality $h(\cdot,\mm)+h''(\cdot,\mm)>0$. One of the advantages of $C^2_+$ bodies is that they are ``stable", in the sense that if $h$ is the support function of a $C^2_+$ body and $\phi$ is in $C^2$, then for sufficiently small $\epsilon$, $h+\epsilon\phi$ is a support function for a $C^2_+$ body \protect\cite{groemer1993perturbations}.

\subsubsection{Zonotopes}
A zonotope is defined to be a Minkowski sum of line segments, and can be interpreted as generalisations of parallelograms (and their higher dimensional analogues), providing a subclass of convex polygonal shapes. Explicitly, if $V=(v_i)_{i=1}^k$ is a list of vectors in $\mathbb{R}^n$, then we define a zonotope 
\begin{equation}\begin{split}
Z(V)
=&\left\{\sum\limits_{i=1}^k a_i v_i : a_i \in \left[-\frac{1}{2},\frac{1}{2}\right]\right\}\\
=&\sum\limits_{i=1}^k \left\{a_i v_i : a_i \in \left[-\frac{1}{2},\frac{1}{2}\right]\right\}.
\end{split}\end{equation} 
For a zonotope in $\mathbb{R}^2$, its volume can be readily computed as $|Z(V)|=\sum\limits_{1\leq i < j \leq k} |v_i \times v_j|$, where $v_i \times v_j = |v_i||v_j|\sin(\theta)$, and $\theta$ is their relative angle. The perimeter is given by $\mathrm{Per}(Z(V))=2\sum\limits_{i=1}^k |v_i|$. Furthermore, if $V_1,V_2$ are lists of vectors, and $V$ is their concatenation, then $Z(V_1)+Z(V_2)=Z(V)$. Zonotopes are of particular interest in the study of excluded volumes, as they admit simpler and explicit methods of computation \protect\cite{mulder2005excluded}.

The closure of the set of zonotopes with respect to the Hausdorff metric is called the set of {\it zonoids}. In two dimensions this simply reduces to the set of convex bodies centred at $0$ such that $-\mm=\mm$. This translates to the requirement that their support functions are $\pi$-periodic, i.e. $h(\theta,\mm)=h(\theta+\pi,\mm)$. We denote by $\mathcal{Z}$ the set of zonoids.

\subsection{Outline of the paper}

In \Cref{secNonUniqueness}, the main result is to show that if $\mm$ is a $C^2_+$ body which admits a support function with more than two non-zero Fourier coefficients, there exists perturbations of the body, not equivalent modulo rigid motions, with the same excluded area function (\Cref{theoremNonUniqueness}). Simply put, for a large class of bodies we cannot uniquely reconstruct a body given its excluded area function. The proof strategy is to obtain a relatively simple expression for the excluded area function in terms of the Fourier coefficients of the support function (\Cref{lemEqVolume}), which gives an {\it ad-hoc} method of finding the Fourier coefficients of the excluded area function, from which the results can be drawn. The proof strategy also proves that if two shapes admit the same excluded area function, then the Fourier coefficients of their support functions may differ only in phase, but not magnitude. 

In spite of this rather general non-uniqueness result, in \Cref{secReconstruct} we turn to the problem of reconstructing a convex body given only its excluded area function. The results are asymptotic, in that we consider a discretised scenario of investigating zontopes with a fixed number of spanning vectors, and derive an algebraic least-squares-type minimisation problem on the spanning vectors. Explicitly, the finite dimensional minimisation problem is equivalent to minimising 
\begin{equation}\begin{split}
||P_M(f-Ex(\cdot,Z(V)))||_2^2,
\end{split}\end{equation}
where $P_M$ is the projection onto the first $M$ Fourier coefficients, $f$ is the trial excluded area function, and the minimisation takes place over all lists of vectors $V$ of at most $k$ vectors. This gives $k$ and $M$ as our discretisation parameters.

It is shown that, at least up to subsequences, solutions converge in the Hausdorff sense to some zonoid, and the excluded area functions corresponding to solutions converge in a sense stronger than strong-$W^{1,p}$ convergence ($p<\infty$) (\Cref{theoremConverge}). It is not necessary that $f$ actually be an excluded area function. If $f$ is an excluded area function for a zonoid, then we establish a convergence rate of solutions $Ex(\cdot,Z(V_{k,M}))$ (\Cref{theoremConvergeRate}). In the case where $f$ is not an excluded area function for a zonoid, it is shown that the solutions $Ex(\cdot,Z(V_{k,M}))$ converge in a sense stronger than strong-$W^{1,p}$ ($p<\infty$) to a best $L^2$ approximation of $f$ in the space of excluded area functions of zonoids. The nature of the discretisation used is that the reconstruction algorithm is only able to reconstruct zonoids.

In \Cref{subsecExamples}, examples of solutions obtained from the reconstruction algorithm are shown. The implementation was done in Wolfram Mathematica using standard in-built numerical minimisation procedures. These examples show qualitatively what is expected from the analysis, and the implementation was able to provide satisfactory results on roughly 10 minute runs on a desktop at the most refined discretisation, with nearly zero code optimisation. 

In \Cref{secOnsagerComparison}, we return to the motivating case of the Onsager free energy. Throughout the previous analysis we consider convergence of convex bodies in Hausdorff metric as the ``natural" method of convergence, and this section is devoted to verifying that Hausdorff convergence, is in a sense, compatible with Onsager. To this end, for arbitrary convex bodies in $\mathbb{R}^n$, we show that convergence of bodies in Hausdorff metric leads to convergence of the corresponding minimisers of Onsager's model by techniques of $\Gamma$-convergence, with the precise statements in \Cref{theoremOnsagerConverge}.

\section{The excluded area and (non)-uniqueness}\label{secNonUniqueness}

In this section we will prove that for a rather large class of convex bodies in $\mathbb{R}^2$, we cannot expect $\mm$ to be uniquely constructed given $Ex(\cdot,\mm)$. For the most part we show non-uniqueness by using appropriate $C^2$ perturbations of the support function, and for this reason our exact non-uniqueness result is restricted to $C^2_+$ bodies, as they are stable under such perturbations. Before proceeding, we need some general results on the excluded area function and its relationship to the support function, for which the Fourier decomposition provides great insight. Throughout, when $\mm$ is unambiguous we will write $h(\theta)=h(\theta,\mm)$.

\begin{proposition}[\protect\cite{ghosh1998support}]\label{propVol}
If the support function of a convex body $\mm$ is given as $h(\theta)=\sum\limits_{n=-\infty}^\infty c_n e^{in\theta}$, where $c_n =\frac{1}{2\pi}\int_{-\pi}^\pi h(\theta)e^{-in\theta}\,d\theta$, then 
\begin{equation}\begin{split}
|\mm|= \pi\sum\limits (1-n^2)|c_n|^2.
\end{split}\end{equation}
\end{proposition}

We note that there is no contribution from $c_1,c_{-1}$ as these amount only to translations of the body and thus do not alter the volume. Furthermore, the fact that the phase of the Fourier coefficients does not alter the volume is evocative to the fact that volume is rotation invariant.

It will be useful to reference the real-space analogue of this formula, which for $C^2$ support functions is given as 
\begin{equation}\begin{split}
|\mm|=\frac{1}{2}\int_0^{2\pi} h(\theta)\big(h(\theta)+h''(\theta)\big)\,d\theta
\end{split}\end{equation}

This representation allows us to provide an expression for the excluded area function relatively quickly

\begin{lemma}\label{lemEqVolume}
If $h$ is the support function of a convex body $\mm$, and has a Fourier decomposition $h(\theta,\mm)=\sum\limits_{n=-\infty}^\infty c_ne^{in\theta}$, then 
\begin{equation}\begin{split}\label{eqVolume}
&Ex(\theta,\mm)=
2\pi\sum\limits_{n=-\infty}^\infty (1-n^2)\big(1+(-1)^n\cos(n\theta)\big)\left|c_n\right|^2
\end{split}\end{equation}
\end{lemma}
\begin{proof}
First, we recall that $h(\theta,-\mm)=h(\theta+\pi,\mm)$. Furthermore, Minkowski sums distribute over support functions, in the sense that $h(\cdot,\mm_1+\mm_2)=h(\cdot,\mm_1)+h(\cdot,\mm_2)$. This means we can find the Fourier coefficients of $h(\cdot,\mm-R_{\omega}\mm)$ as 
\begin{equation}\begin{split}
&\frac{1}{2\pi}\int_{-\pi}^\pi h(\omega,\mm-R_{\omega}\mm)e^{-in\omega}\,d\omega
=\frac{1}{2\pi}\int_{-\pi}^\pi e^{-in\omega}h(\omega,\mm)
+e^{-in\omega}h(\omega+\pi,R_{\theta}\mm)\,d\omega\\
=&\frac{1}{2\pi}\int_{-\pi}^\pi e^{-in\omega}h(\omega,\mm)
+e^{-in\omega-in\pi}h(\omega-\theta,\mm)\,d\omega\\
=& \frac{1}{2\pi}\int_{-\pi}^\pi e^{-in\omega} h(\omega,\mm)\,d\omega
+ \frac{1}{2\pi}\int_{-\pi}^\pi e^{-in\omega-in\pi}h(\omega-\theta,\mm)\,d\omega\\
=&c_n + e^{in\theta-in\pi}c_n=(1+e^{in\theta-in\pi})c_n.
\end{split}\end{equation}
Using the area formula using Fourier coefficients in \Cref{propVol}, we then have 
\begin{equation}\begin{split}
Ex(\theta,\mm)=& \pi\sum\limits_{n=-\infty}^\infty (1-n^2)\left|(1+e^{in\theta-in\pi})c_n\right|^2\\
=& \pi\sum\limits_{n=-\infty}^\infty (1-n^2)\left|1+e^{in\theta-in\pi}\right|^2\left|c_n\right|^2\\
=&  2\pi\sum\limits_{n=-\infty}^\infty (1-n^2)\big(1+(-1)^n\cos(n\theta)\big)\left|c_n\right|^2
\end{split}\end{equation}
\end{proof}

Using the real-space representation we can obtain a further expression for the excluded volume function. 

\begin{proposition}
Let $\mm$ be a convex body with $C^2$ support function. Then 
\begin{equation}\begin{split}
&Ex(\theta,\mm)= 2|\mm|+\frac{1}{2}\int_0^{2\pi}\big(h(\omega+\theta+\pi)+h(\omega-\theta+\pi)\big) \big(h(\omega)+h''(\omega)\big)\,d\omega.
\end{split}\end{equation}
\end{proposition}
\begin{proof}
Recall that 
\begin{equation}\begin{split}
|\mm|=\frac{1}{2}\int_{0}^{2\pi}h(\omega,\mm)\big(h(\omega,\mm)+h''(\omega,\mm)\big)\,d\omega.
\end{split}\end{equation}
For brevity let $h=h(\cdot,\mm)$. This then gives that 
\begin{equation}
\begin{split}
 &Ex(\theta,\mm)\\
 =&|\mm-R_\theta\mm|\\
 =& \frac{1}{2}\int_{0}^{2\pi}\big(h(\omega,\mm-R_\theta\mm)\big) \big(h(\omega,\mm-R_\theta\mm)+h''(\omega,\mm-R_\theta\mm)\big)\,d\omega\\
 =& \frac{1}{2}\int_{0}^{2\pi}\big(h(\omega)+h(\omega-\theta+\pi)\big)\big(h(\omega)+h(\omega+\pi-\theta)+h''(\omega)+h''(\omega+\pi-\theta)\big)\,d\omega\\
 =& 2|\mm|+\frac{1}{2}\int_0^{2\pi}h(\omega)\big(h(\omega+\pi-\theta)+h''(\omega+\pi-\theta)\big)\,d\omega+\frac{1}{2}\int_{0}^{2\pi}h(\omega-\theta+\pi)\big(h(\omega)+h''(\omega)\big)\,d\omega\\
 =& 2|\mm|+\frac{1}{2}\int_0^{2\pi}\big(h(\omega+\theta-\pi)+h(\omega-\theta+\pi)\big)\big(h(\omega)+h''(\omega)\big)\,d\omega\\
 =& 2|\mm|+\frac{1}{2}\int_0^{2\pi}\big(h(\omega+\theta+\pi)+h(\omega-\theta+\pi)\big)\big(h(\omega)+h''(\omega)\big)\,d\omega
\end{split} 
\end{equation}
where the fact that $h$ is $2\pi$-periodic is used in the last line. 
\end{proof}

This gives an immediate corollary, analogous to results in three dimensions about the minimum excluded volume of achiral, axially symmetric convex bodies \protect\cite{palffy2014minimum}.

\begin{corollary}
For all convex bodies in $\mathbb{R}^2$, the excluded area function is minimised at an anti-parallel configuration. 
\end{corollary}
\begin{proof}
First rewrite \ref{eqVolume} as
\begin{equation}\begin{split}
Ex(\theta,\mm)= 2\pi\sum\limits_{n=-\infty}^\infty (1-n^2)\big(1+\cos(n(\theta+\pi))\big)\left|c_n\right|^2.
\end{split}\end{equation}
The $\theta$ dependent terms appear as $\cos(n(\theta+\pi))$ multiplied by negative constants. Thus if $\theta=\pi$, corresponding to an anti-parallel configuration, the excluded volume is minimised. 
\end{proof}

\begin{remark}
The formula gives us a relationship between the Fourier decay rates, and subsequently the regularity (in an $H^s$ fractional Sobolev sense), of the support function $h(\cdot,\mm)$ and the excluded area function $Ex(\cdot,\mm)$. We see that if $c_n\sim n^{-s}$, then the Fourier coefficients of $Ex(\cdot,\mm)$ decay as $n^{2(1-s)}$, and vice versa. 
\end{remark}

\begin{remark}
We see that the expression for $Ex(\cdot,\mm)$ has no contribution from the first order Fourier coefficients of the support function. This is due to the fact that altering the first order Fourier coefficients amounts to translating the body. This lack of dependence is evocative to work of Piastra and Virga, which in three dimensions demonstrates a lack of dipolar component in the excluded volume of axially symmetric convex bodies \protect\cite{piastra2015explicit}. 
\end{remark}

While the majority of this work will be devoted to single-component systems, the result \Cref{lemEqVolume} can readily be extended to two-component pairwise excluded volumes. 
\begin{proposition}\label{propTwoComp}
Let $\mm_1,\mm_2$ be convex bodies. Their pairwise excluded volume, at relative angle $\theta$, is given by 
\begin{equation}\begin{split}
  &Ex(\theta;\mm_1,\mm_2)\\
 =&|\mm_1|+|\mm_2|+2\pi\sum\limits_{n=-\infty}^\infty(1-n^2)\left(\Re\left(c_{1,n}^*c_{2,n}\right)\cos n(\theta+\pi)+\Im\left(c_{1,n}^*c_{2,n}\right)\sin n(\theta+\pi)\right)
\end{split}\end{equation}
\end{proposition}
\begin{proof}
By \protect\cite{mulder2005excluded}, the excluded volume is given by $Ex(\theta;\mm_1,\mm_2)=|\mm_1-R_{\theta}\mm_2|$. By the same means, we employ the support functions $h_j(\theta)=h(\theta,\mm_j)$ with respective Fourier decompositions $h_j(\theta)=\sum\limits_{n=-\infty}^\infty c_{j,n}e^{in\theta}$. Then, following the same methodology as \Cref{lemEqVolume},
\begin{equation}\begin{split}
 & Ex(\theta;\mm_1,\mm_2)\\
 =& \pi\sum\limits_{n=-\infty}^\infty(1-n^2)\big|c_{1,n}+e^{in(\theta+\pi)}c_{2,n}\big|^2\\
 =& \pi\sum\limits_{n=-\infty}^\infty(1-n^2)\left(\big|c_{1,n}\big|^2+\big|c_{2,n}|^2+2\Re\left(c_{1,n}^*e^{in(\theta+\pi)}c_{2,n}\right)\right)\\
 =&|\mm_1|+|\mm_2|+2\pi\sum\limits_{n=-\infty}^\infty(1-n^2)\left(\Re\left(c_{1,n}^*c_{2,n}\right)\cos n(\theta+\pi)+\Im\left(c_{1,n}^*c_{2,n}\right)\sin n(\theta+\pi)\right)
\end{split}\end{equation}
\end{proof}

\begin{remark}
A remarkable consequence is that the equation given by \Cref{propTwoComp} may be independent of $\theta$ for non-trivial shapes. Provided at least one of $c_{1,n}$, $c_{2,n}$ is zero for all $n\neq 0,1,-1$, we have that $Ex(\theta;\mm_1,\mm_2)$ is a constant. If we consider an Onsager model for such a two component system, the free energy density is
\begin{equation}\begin{split}
\mf=& \int_{\mathbb{S}^1}\rho_1f_1(\theta)\ln f_1(\theta) +\rho_2f_2(\theta)\ln f_2(\theta)\,d\theta+\frac{\rho_1^2}{2}\int_{\mathbb{S}^1}\int_{\mathbb{S}^1}f_1(\theta)f_1(\theta')Ex(\theta-\theta',\mm_1)\,d\theta\,d\theta'\\
&+\rho_1\rho_2\int_{\mathbb{S}^1}\int_{\mathbb{S}^1}f_1(\theta)f_2(\theta')Ex(\theta-\theta',\mm_1,\mm_2)\,d\theta\,d\theta'+\frac{\rho_2^2}{2}\int_{\mathbb{S}^1}\int_{\mathbb{S}^1}f_2(\theta)f_2(\theta')Ex(\theta-\theta',\mm_2)\,d\theta\,d\theta',\\
\end{split}\end{equation}
where $\rho_j,f_j$ are the number densities and one-particle distribution function for the species described by convex body $\mm_j$. If $Ex(\cdot;\mm_1,\mm_2)=C$, a constant, this means that the Onsager free energy fully decouples the orientation distribution function, and the energy reduces to 
\begin{equation}\begin{split}
\mf-C\rho_1\rho_2 =& \int_{\mathbb{S}^1}\rho_1f_1(\theta)\ln f_1(\theta)+\frac{\rho_1^2}{2}\int_{\mathbb{S}^1}\int_{\mathbb{S}^1}f_1(\theta)f_1(\theta')Ex(\theta-\theta',\mm_1)\,d\theta\,d\theta'\\ 
  &+\int_{\mathbb{S}^1}\rho_2f_2(\theta)\ln f_2(\theta)\,d\theta +\frac{\rho_2^2}{2}\int_{\mathbb{S}^1}\int_{\mathbb{S}^1}f_2(\theta)f_2(\theta')Ex(\theta-\theta',\mm_2)\,d\theta\,d\theta'.
\end{split}\end{equation}
This would seem to imply a decoupling in the second order virial coefficient of the system, and it may be possible this degeneracy gives rise to interesting behaviour in real systems, analogous to the degeneracy in the second virial coefficient at the Boyle temperature, though we leave such an investigation open as an avenue for future work. 
\end{remark}

Now we turn to proving the non-uniqueness of the excluded area function for $C^2_+$ bodies. The heuristic is that the expression \Cref{lemEqVolume} essentially gives the Fourier decomposition of $Ex(\cdot,\mm)$, which shows the expression depends only on the norm of the Fourier coefficients of $h(\cdot,\mm)$. Thus if two bodies have the same excluded area function they must have Fourier coefficients equal in absolute value, but my vary in phase. If $h(\cdot,\mm)$ has two non-zero Fourier coefficients (excluding zero and first order coefficients), changing their relative phase by a small amount gives a $C^2$ perturbation of the support function, and thus defines a new $C^2_+$ body for sufficiently small perturbations which can be shown is not a rigid motion of $\mm$. 

\begin{theorem}\label{theoremNonUniqueness}
Let $\mm,\tilde\mm$ be $C^2_+$ bodies with support functions $h,\tilde h$ respectively, with Fourier coefficients $c_n,\tilde c_n$ respectively. Then $Ex(\cdot,\mm)=Ex(\cdot,\tilde\mm)$ if and only if $|c_n|=|\tilde c_n|$ for all $n\neq \pm 1$. In particular, given $\mm\in C^2_+$, there exists a convex body $\tilde\mm$, which is not a rigid motion of $\mm$, so that $Ex(\cdot,\mm)=Ex(\cdot,\tilde\mm)$ if and only if the support function of $\mm$, $h(\cdot,\mm)$, has at least two non-zero Fourier coefficients, not including the constant or first order coefficients.
\end{theorem}
\begin{proof}
First assume $Ex(\theta,\mm)=Ex(\theta,\tilde\mm)$, and the Fourier coefficients of $h(\cdot,\mm)$ and $h(\cdot,\tilde\mm)$ are given by $c_n,\tilde c_n$ respectively. Then by integrating \ref{eqVolume} against $\cos(n\theta)$, we obtain that $|c_n|=|\tilde c_n|$ for all $n\neq 1,0,-1$. We note that changing $c_1$ merely reduces to translating $\mm$. As $|c_n|=|\tilde c_n|$,
\begin{equation}\begin{split}
&0\\
=& Ex(\pi,\mm)-Ex(\pi,\tilde{\mm})\\
=& 2\pi\sum\limits_{n=-\infty}^\infty (1-n^2)\big(1+(-1)^n\cos(n\pi)\big)\left(\left|c_n\right|^2-|\tilde c_n|^2\right)\\
=& 2\pi\sum\limits_{n=-\infty}^\infty (1-n^2)\big(1+1\big)\left(\left|c_n\right|^2-|\tilde c_n|^2\right)\\
=& 4\pi \left(|c_0|^2-|\tilde c_0|^2\right)
\end{split}\end{equation}
 As $c_0,\tilde c_0$ are positive and real, this implies they are equal. So this implies that $c_n=e^{in\xi_n}\tilde c_n$ for all $n\geq 2$ and some real constants $\xi_n$, and for $n\leq -2$ the Fourier coefficients can be found by $c_n=c_{-n}^*$. If $\xi_n=\xi$, independently of $n$ (modulo $2\pi$) for all $c_n\neq 0$, then this implies $h(\theta,\tilde\mm)=h(\theta-\xi,\mm)$, i.e. $\tilde\mm$ is just a rotation of $\mm$. Thus for $\tilde\mm$ to be distinct from $\mm$, modulo rotations, this means that we must have some $n_1,n_2$, where $n_1>n_2\geq 2$ with $\xi_{n_1}\neq \xi_{n_2}$ modulo $2\pi$, and $c_{n_1} \neq 0\neq c_{n_2}$. 

Now, assume that the Fourier coefficients of $h(\cdot,\mm)$ have two non-zero coefficients, $c_{n_1},c_{n_2}$ with $n_1>n_2>1$. Then let $\tilde h(\theta)=\sum\limits_{n=-\infty}^\infty \tilde c_ne^{in\theta}$ where $\tilde c_n=c_n$ for $n\neq \pm n_2$, and $\tilde c_{\pm n_2}=e^{\pm i\epsilon}c_{\pm n_2}$. This means that 
\begin{equation}\begin{split}
&\tilde h(\theta)\\
=& h(\theta,\mm)+(-1)^{n_2}\left((e^{i\epsilon}-1)c_{n_2}e^{in_2\theta}+(e^{-i\epsilon}-1)c_{n_2}e^{in_2\theta}\right)\\
=& h(\theta,\mm)+\phi_\epsilon(\theta).
\end{split}\end{equation}
Here \begin{equation}\begin{split}&\phi_\epsilon(\theta)=(-1)^{n_2}\left((e^{i\epsilon}-1)c_{n_2}e^{in_2\theta}+(e^{-i\epsilon}-1)c_{n_2}e^{in_2\theta}\right).\end{split}\end{equation} We can estimate ${||\phi_\epsilon||_\infty \leq 2|c_{n_2}|\,|e^{i\epsilon}-1|}$, and ${||\nabla^2 \phi_\epsilon||_\infty\leq  n_2^2|c_{n_2}|\,|e^{i\epsilon}-1|}$, both of which tend to zero as $\epsilon \to 0$, so by \protect\cite{groemer1993perturbations}, we have that $\tilde h(\theta)$ defines a support function for sufficiently small $\epsilon>0$ for some convex body $\tilde \mm$, and by taking $\epsilon$ irrational, we must have that $\mm'$ is not simply a rotation of $\mm$. Furthermore, from the representation in \Cref{lemEqVolume}, this means that $Ex(\cdot,\mm)=Ex(\cdot,\tilde\mm)$
\end{proof}

\begin{corollary}
In two dimensions, the only convex body $\mm$ so that $Ex(\cdot,\mm)$ is constant is a disc.
\end{corollary}
\begin{proof}
The Fourier decomposition of $h(\theta,B)$ is simply $c_0=r$, and $c_n=0$ otherwise, and a ball is a $C^2_+$ body
\end{proof}

\subsection{Chirality}

We now proceed to make some comments about chirality and its relationship with the excluded area function. It should be noted that the different nature of chirality in odd and even dimensions mean that analogues of these results in three dimensions would have a rather different flavour. In particular, head-to-tail symmetry in $\mathbb{R}^2$ is more of the flavour of achirality in $\mathbb{R}^3$, as both are represented by similar symmetries of the support function, $h(\theta)=h(\theta+\pi)$ and $h(\theta,\phi)=h(\pi-\theta,\phi+\pi)$ in two and three dimensions respectively, in an appropriate coordinate system. Achiral bodies in two dimensions however are, in a sense, more analogous to axially symmetric bodies in three dimensions, as axially symmetric bodies are precisely those obtained from rotating a 2D achiral body about its symmetry line. If we are to consider the 2D scenario as a testing ground for 3D, to this end it should be expected that the following results are more analogous to those we expect for axially symmetric bodies in 3D. 

\begin{corollary}
If $\mm$ is an achiral $C^2_+$ body, and admits two non-zero Fourier coefficients (not including the constant or first order term) then there exists a chiral convex body $\mm'$ so that $Ex(\cdot,\mm)=Ex(\cdot,\mm')$. 
\end{corollary}
\begin{proof}
First we note that if $\mm$ is achiral, then there exists some $\phi$ so that $c_n=\tilde{c}_ne^{in\phi}$, with $\tilde{c}_n$ being real valued. In particular, $\frac{(c_{n_1})^{n_2}}{(c_{n_2})^{n_1}}$ is real valued for all $n_1,n_2$. This corresponds to $h(\phi+\theta,\mm)=h(\phi-\theta,\mm)$ symmetry. Without loss of generality, take $\phi=0$, so that $h(\cdot,\mm)$ admits only real Fourier coefficients. Recalling the construction in \Cref{theoremNonUniqueness}, if $\epsilon$ is sufficiently small, we have that all but one non-zero Fourier coefficient is real. This in particular means that $\frac{(c_{n_1})^{n_2}}{(c_{n_2})^{n_1}}$ must have an imaginary component, as it is the ratio between a real and a fully complex number. Thus $\mm'$ is chiral.
\end{proof}

\begin{remark}
The constructed examples showing non-uniqueness are typically chiral configurations, caused by small perturbations. However it is not necessarily the case, we can show examples of non-uniqueness of the excluded area function within only the class of achiral bodies. A simple example would be to take the convex bodies $\mm^\pm$ with support functions $h^\pm(\theta)=20+\cos(2\theta)\pm \cos(4\theta)$. It is immediate to see that these are support functions, as $h^\pm +(h^\pm)''>2$, each choice $\pm$ gives bodies not equal modulo symmetry, and both choices give achiral bodies with the same excluded area function.
\end{remark}

\subsection{Examples}

\begin{example}
Let $\mm$ be an ellipse with minor and major axes $1$ and $\sqrt{3}$. This has the support function $h(\theta,\mm)=\sqrt{1+2\cos^2\theta}$. As the support function contains only cosine terms in its trigonometric Fourier expansion, we can write such a perturbation of a single coefficient as $\tilde{h}(\theta)=h(\theta)+a_{2n}\left(\cos(2n\theta+\omega)-\cos(2n\theta)\right)$. We perturb the $c_2$ coefficient as this gives a more noticeable difference. We have that $a_2=\frac{1}{\pi}\int_{\mathbb{S}^1}h(\theta)\cos(2\theta)\,d\theta\approx 0.1816$ numerically. In trigonometric notation, the corresponding perturbation of phase in the $c_2$ coefficient corresponds to producing a new support function $\tilde{h}(\theta)=h(\theta)+a_2(\cos(2\theta-\omega)-\cos(2\theta))$. We check the support function condition, that $\tilde{h}+\tilde{h}''\geq 0$, and numerically we see that this holds at least for $\omega \in \left[0,\frac{21\pi}{50}\right]$, and $\omega \in\left[\frac{18\pi}{25},\pi\right]$, while failing for some interval between them. These are not exact values, but loose upper bounds. In this case we can produce several chiral, convex bodies which have the same excluded area function as an ellipse. Constructions from the support function are given in \Cref{figEllipses}. The difference is not significant by eye-norm, to be expected for perturbations of a single body. We compare the difference quantitatively by defining a metric modulo symmetry, of $d(\mm,{\mm}')=\min\limits_{R\in\mbox{SO}(2)}d_H(\mm,R{\mm'})=\min\limits_{\omega}\max\limits_{\theta}|h(\theta,\mm)-h(\omega,{\mm'})|$, for which we include the numerically approximated values.

\begin{figure*}
\begin{subfigure}{0.3\textwidth}
\includegraphics[width=\textwidth]{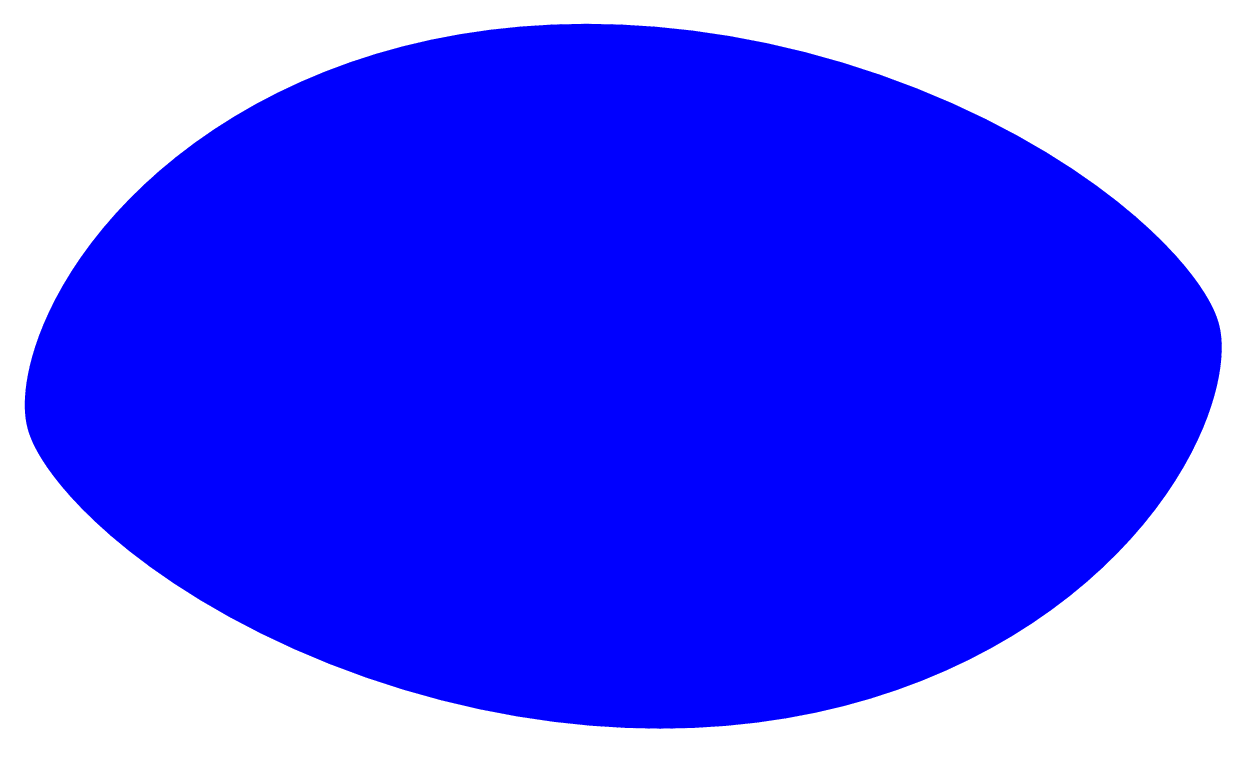}
\caption{$\omega=\frac{\pi}{4}$,  $d(\mm,\mm')\approx 0.019$}
\end{subfigure}
\begin{subfigure}{0.3\textwidth}
\includegraphics[width=\textwidth]{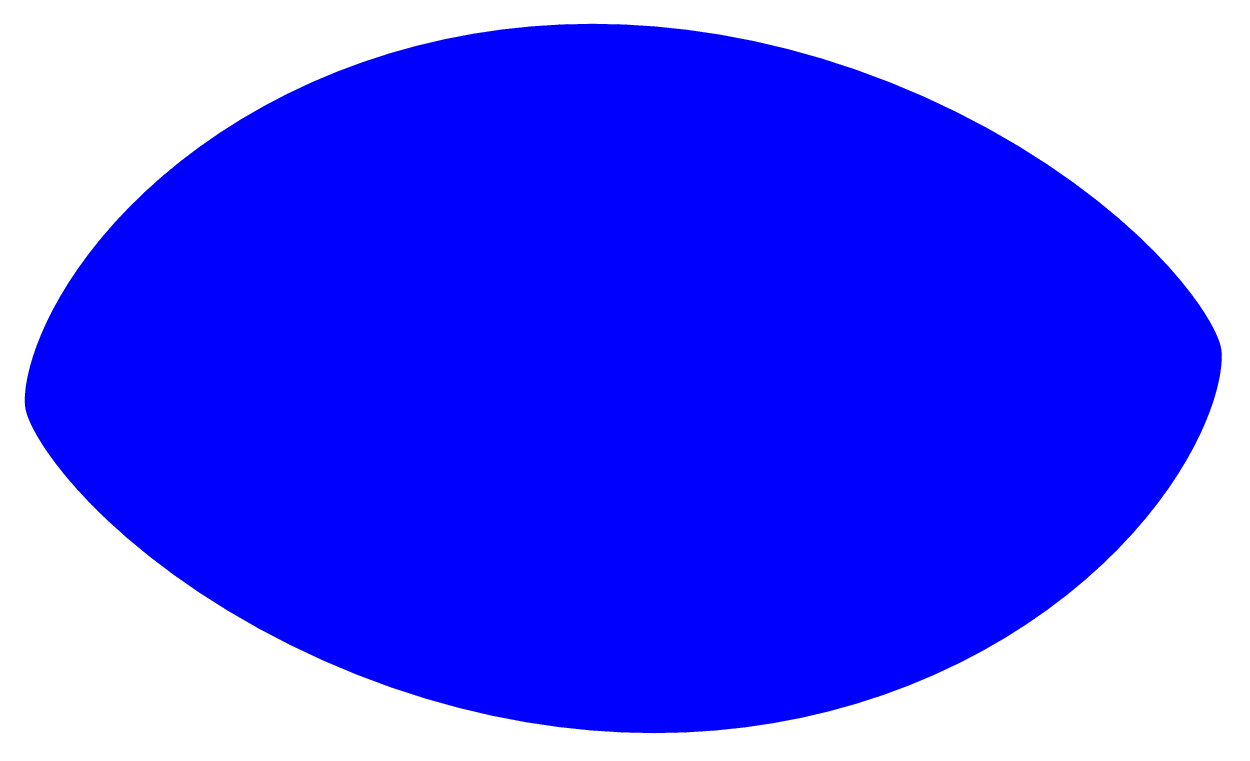}
\caption{$\omega=\frac{\pi}{3}$, $d(\mm,\mm')\approx 0.023$}
\end{subfigure}
\begin{subfigure}{0.3\textwidth}
\includegraphics[width=\textwidth]{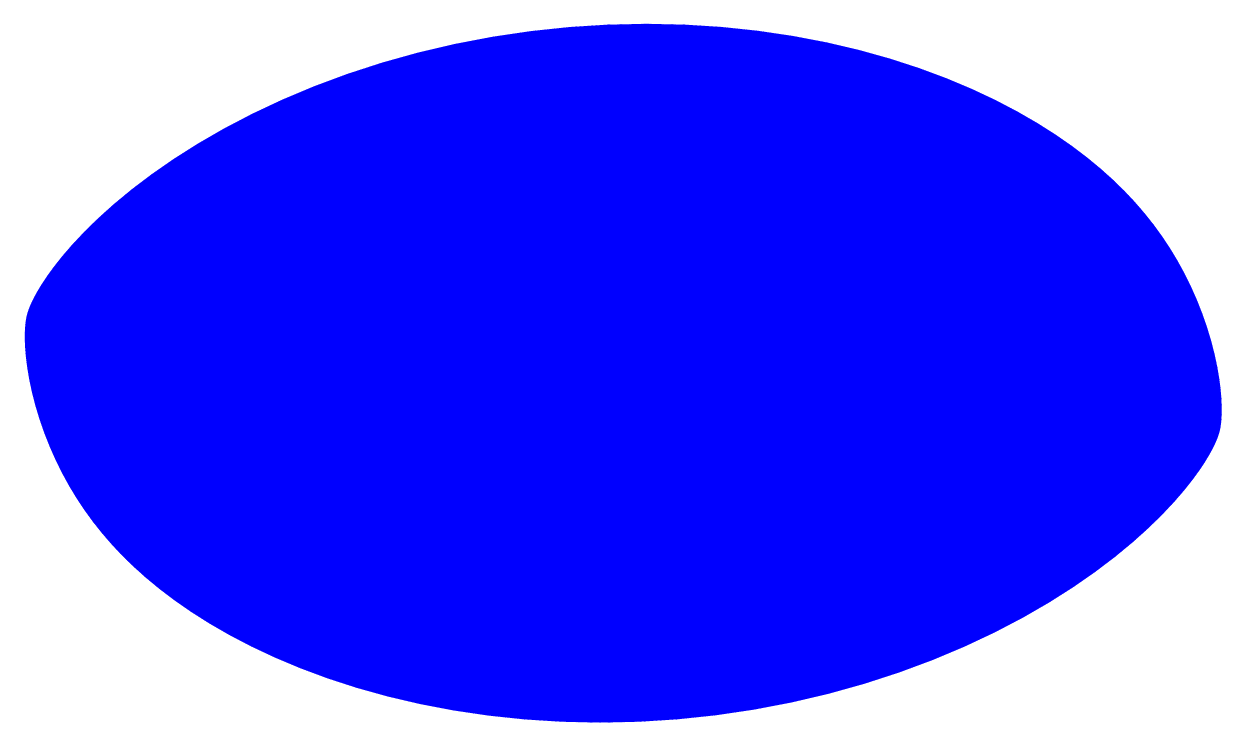}
\caption{$\omega=\frac{5\pi}{6}$, $d(\mm,\mm')\approx 0.013$}
\end{subfigure}
\caption{Convex bodies, described as perturbations of an ellipse, that all have the same excluded area function. Note that all are chiral.}\label{figEllipses}
\end{figure*}
\end{example}

\begin{example}
Consider the family of support functions, indexed by parameter $\omega\in\mathbb{R}$, defined by $h_\omega(\theta)=1+\frac{1}{10}\cos(2\theta)+\frac{1}{25}\cos 4(\theta-\omega)$. These are support functions for convex bodies as $h_\omega+h_\omega''\geq \frac{1}{10}$. Let $\mm_\omega$ be the convex body so $h_\omega=h(\cdot,\mm_\omega)$. 

In exponential form, the Fourier coefficients of $h_\omega$ are given by $c_{2}=c_{-2}=\frac{1}{20}$, $c_{4}=c_{-4}=\frac{1}{50}e^{4i\omega}$. In particular this means all of the convex bodies admit the same excluded area function, given by 
\begin{equation}\begin{split}
Ex(\theta,{\mm}_\omega)=\pi\left(\frac{1919}{1000}-\frac{3}{20}\cos 2\theta-\frac{3}{250}\cos 4\theta\right)
\end{split}\end{equation} 
We show several representative constructions of $\mm_\omega$ within our continuous family in \Cref{figRotated}

\begin{figure*}
\begin{subfigure}{0.18\textwidth}
\includegraphics[width=0.9\textwidth]{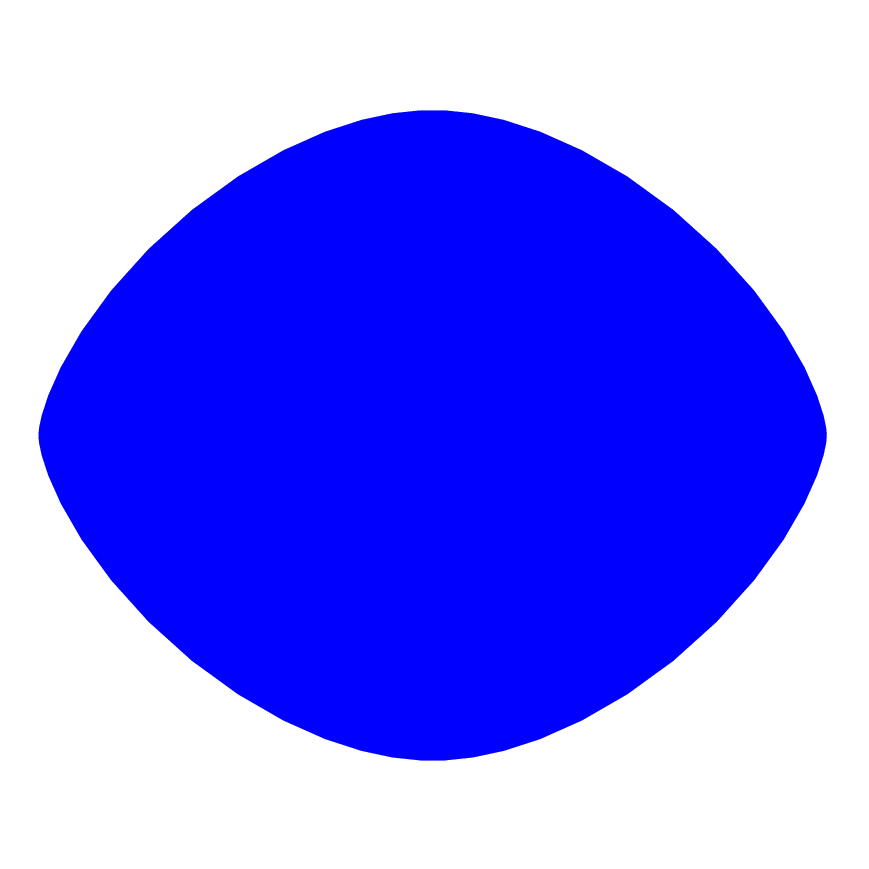}
\caption{$\omega=0$}
\end{subfigure}
\begin{subfigure}{0.18\textwidth}
\includegraphics[width=0.9\textwidth]{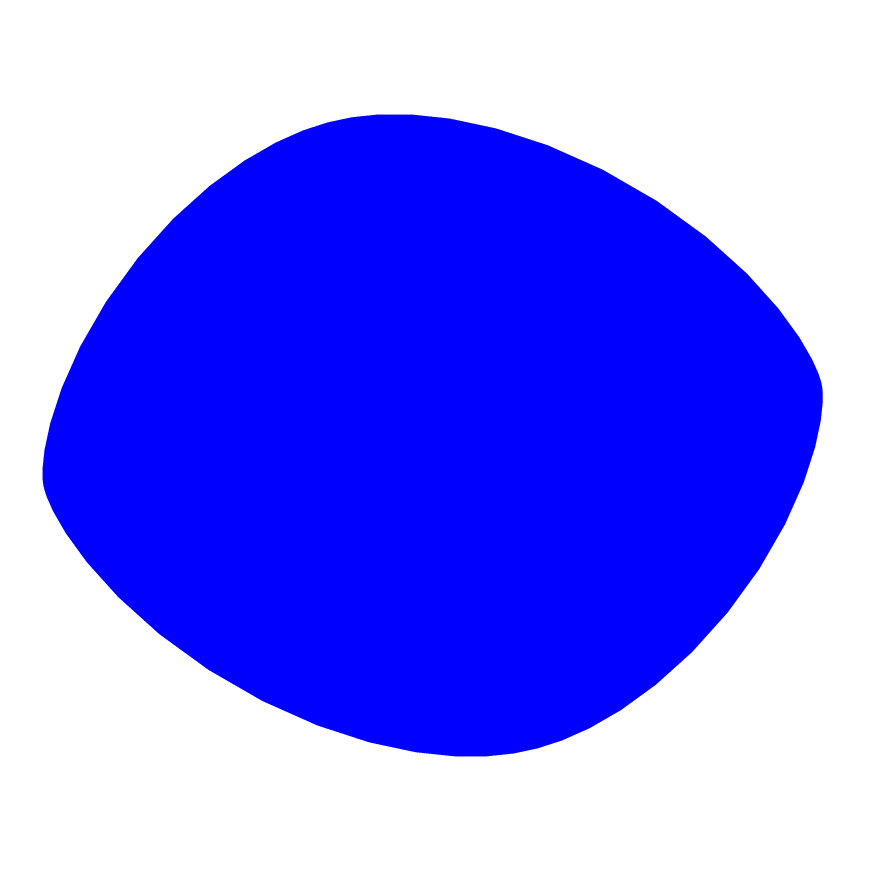}
\caption{$\omega=\frac{\pi}{4}$}
\end{subfigure}
\begin{subfigure}{0.18\textwidth}
\includegraphics[width=0.9\textwidth]{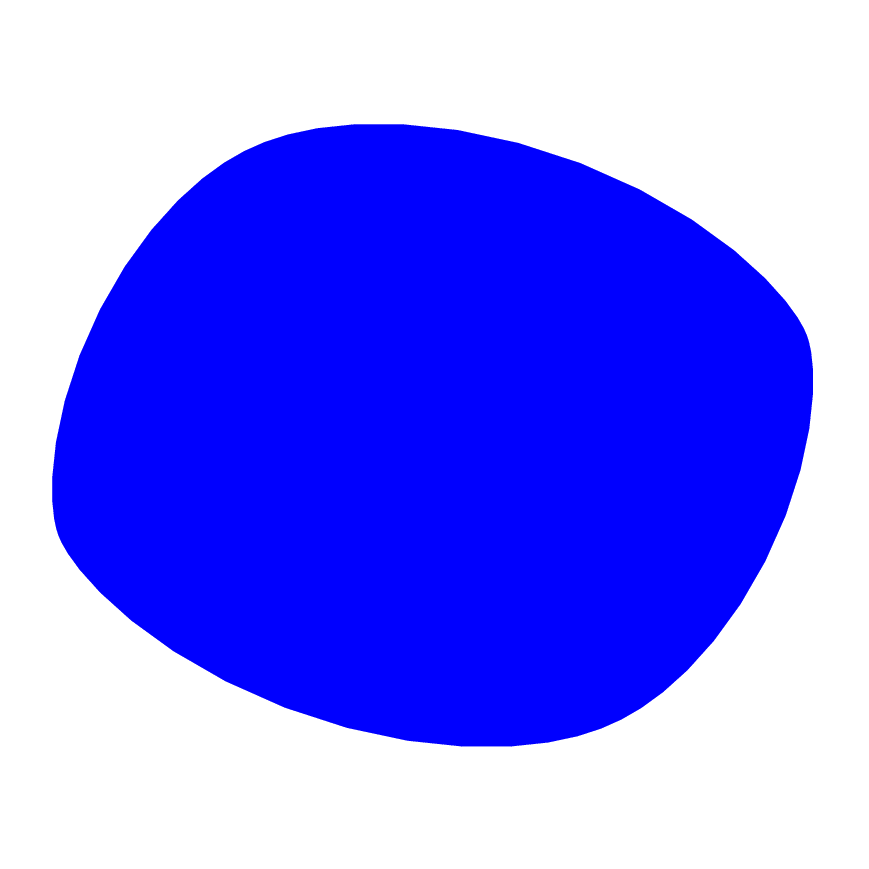}
\caption{$\omega=\frac{\pi}{2}$}
\end{subfigure}
\begin{subfigure}{0.18\textwidth}
\includegraphics[width=0.9\textwidth]{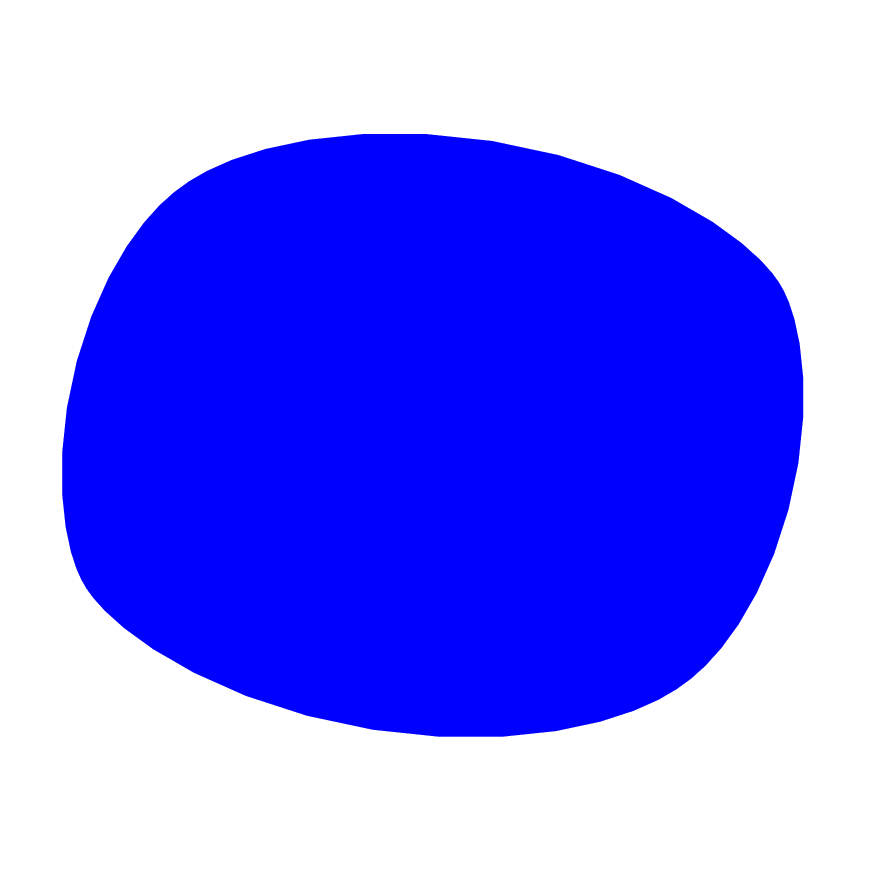}
\caption{$\omega=\frac{3\pi}{4}$}
\end{subfigure}
\begin{subfigure}{0.18\textwidth}
\includegraphics[width=0.9\textwidth]{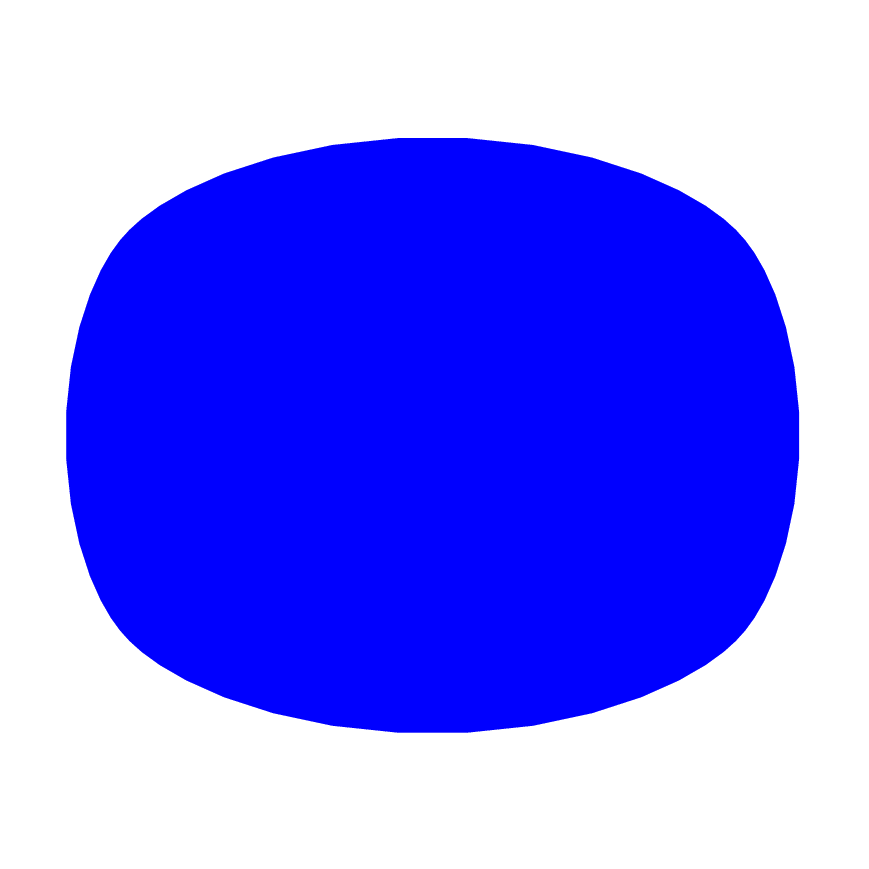}
\caption{$\omega=\pi$}
\end{subfigure}
\caption{$\mm_\omega$ for representative values of $\omega$}\label{figRotated}
\end{figure*}
\end{example}

\begin{remark}
This example highlights one interpretation of how these different functions with same excluded area function can be related. In the case here, we can write $h_\omega(\theta)=\left(\frac{7}{20}-\frac{1}{10}\cos 2\theta\right)+\left(\frac{13}{20}-\frac{1}{25}\cos 4(\theta-\omega)\right)=h_1(\theta)+h_2(\theta-\omega)$. It is readily seen that $h_1$ and $h_2$, given by the left and right bracketed terms, both satisfy the condition to be support functions for a convex body, $\mm_1,\mm_2$ respectively. By the properties of support functions this implies that we can ``decompose" $\mm_\omega = \mm_1+R_\omega\mm_2$. In words, the shapes we constructed with the same excluded area function correspond to rotating each element of a ``basis expansion" (in a loosely defined sense) of $\mm$ separately. This is why the family of convex bodies $\mm_\omega$, by eye-norm, have a rotation-like appearance. We show the ``basis" bodies, and the excluded area function shared by all of them, in \Cref{figBasis}

\begin{figure*}
\begin{center}
\begin{subfigure}{0.25\textwidth}\begin{center}
\includegraphics[width=0.9\textwidth]{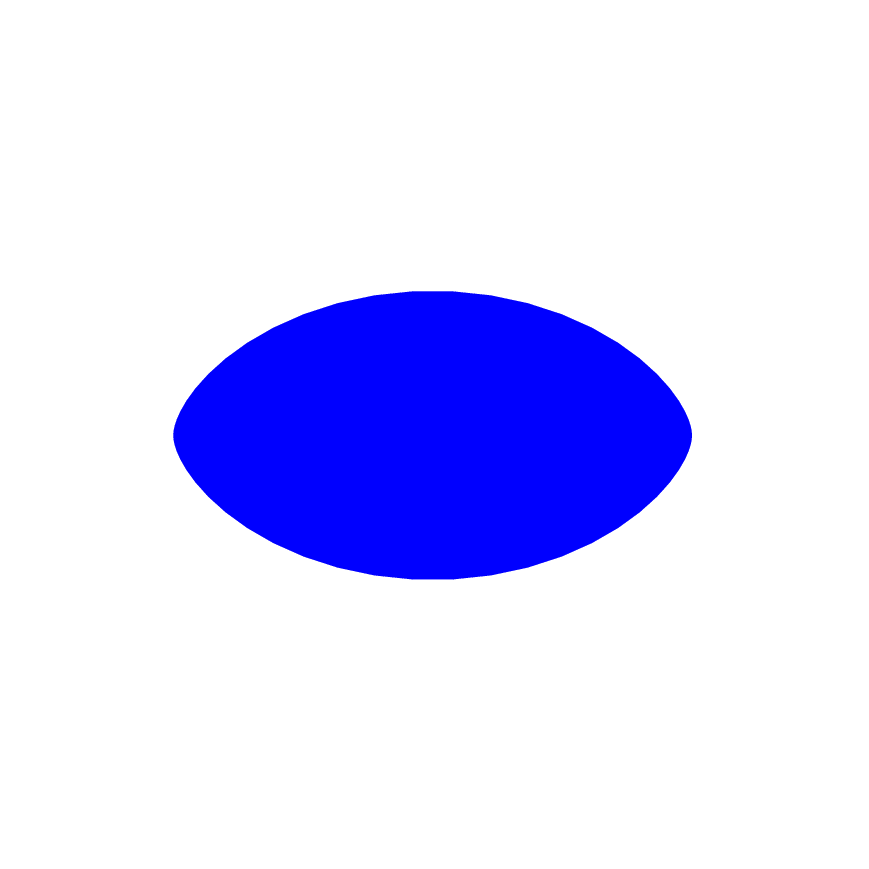}
\caption{$\mm_1$}
\end{center}
\end{subfigure}
\begin{subfigure}{0.25\textwidth}\begin{center}
\includegraphics[width=0.9\textwidth]{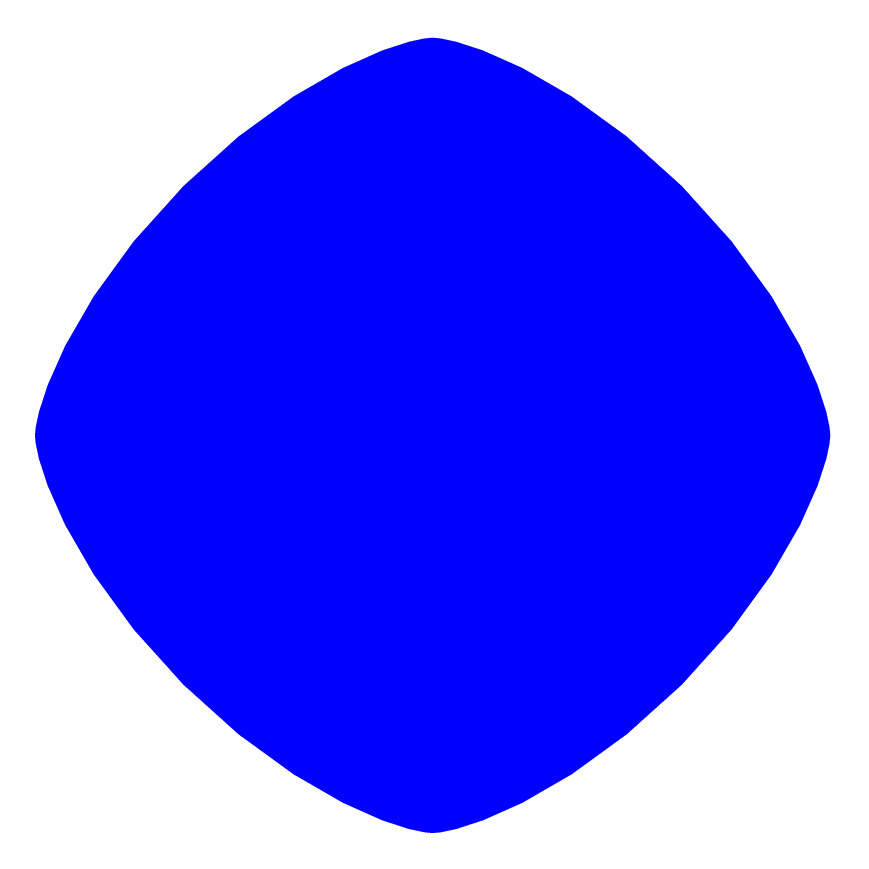}
\caption{$\mm_2$}\end{center}
\end{subfigure}\hspace{0.05\textwidth}
\begin{subfigure}{0.4\textwidth}\begin{center}
\includegraphics[width=0.9\textwidth]{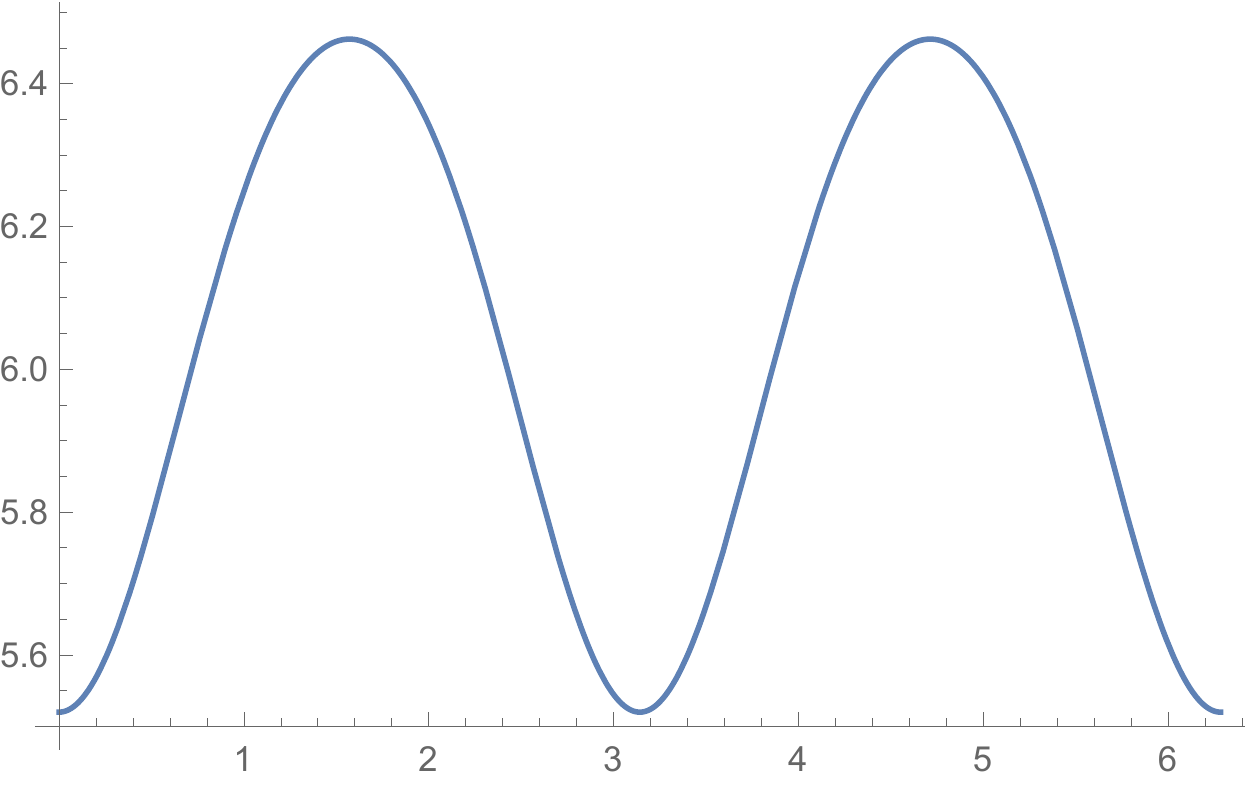}
\caption{$Ex(\cdot,\mm_\omega)$}\end{center}
\end{subfigure}
\end{center}
\caption{The two ``basis" shapes $\mm_1,\mm_2$ so that $\mm_\omega = \mm_1+R_{\frac{\omega}{4}}\mm_2$ in our constructions, shown to scale. The bodies are shown with the excluded area function of their Minkowski sums $\mm_\omega$}\label{figBasis}
\end{figure*}
\end{remark}

\section{Reconstruction algorithm}\label{secReconstruct}

The question can be asked, given the excluded area function $Ex(\cdot,\mm)$, is it possible to reconstruct $\mm$? As we have seen in \Cref{secNonUniqueness}, we can not expect this as non-uniqueness makes the question ill-posed. In this section we obtain perhaps the best possible result given that the question is ill-posed. The key result is that given a trial excluded area function $f$, we can generate a sequence of converging zonotopes $Z(V_k)$ and some zonoid $\mm$ with $Z(V_k)\to\mm$ and $Ex(\cdot,Z(V_k))$ converges in weak-* $W^{1,\infty}$ and strongly in $W^{1,p}$ for $p<\infty$ to a best $L^2$ approximation of $f$ within the space of all excluded area functions of zonoids. In particular, if $f=Ex(\cdot,\mm')$ for some zonoid $\mm'$, then $Ex(\cdot,\mm')=Ex(\cdot,\mm)$. In light of our non-uniqueness result, we cannot however expect $\mm'=\mm$, outside of some very limited cases.

\subsection{Derivation}

\begin{proposition}\label{propVolZono}
Let $\mm=Z(V)$ be a zonotope, where $V=(v_i)_{i=1}^k$, and $v_i = l_iR_{\theta_i}e_1$ for lengths $l_i\geq 0$ and angles $\theta_i$. Then $Ex(\theta,\mm)=b^0(V)-\sum\limits_{m=0}^\infty b^m(V)\cos(2m\theta)$, where the coefficients $b^m(V)$ satisfy 
\begin{equation}\begin{split}
\label{eqZonoFourier}
b^0(V)=& \frac{1}{2\pi}\int_0^{2\pi}Ex(\theta,\mm(V))\,d\theta \\
=& 2\sum\limits_{1\leq i<j\leq k}l_il_j|\sin(\theta_i-\theta_j)|+\frac{2}{\pi}\left(\sum_{i=1}^k l_i\right)^2,\\
b^m(V)=& \sum\limits_{i,j=1}^k\frac{4l_il_j\cos(2m(\theta_i-\theta_j))}{\pi(4m^2-1)}\\
=&\frac{1}{\pi}\int_{0}^{2\pi}\cos(2m\theta)Ex(\theta,\mm(V))\,d\theta.
\end{split}\end{equation}
\end{proposition}
\begin{proof}
First we note that for any zonotope, we have that 
\begin{equation}\begin{split}
 &Ex(\theta,\mm(V))\\
 =&|\mm(V)+ R_\theta\mm(V)|\\
 =&\sum\limits_{1\leq i<j\leq k}|v_i\times v_j| + \sum\limits_{1\leq i<j\leq k}|R_\theta v_i \times R_\theta v_j| + \sum\limits_{i=1}^k\sum\limits_{j=1}^k |v_i\times R_\theta v_j|\\
 =& 2|\mm(V)| + \sum\limits_{i,j=1}^k|v_i \times R_\theta v_j|.
\end{split}\end{equation}
Using the Fourier representation of \begin{equation}\begin{split}|\sin(x)|=\frac{2}{\pi}-\frac{4}{\pi}\sum\limits_{m=1}^\infty \frac{\cos(2mx)}{4m^2-1},\end{split}\end{equation} we can write this as 
\begin{equation}\begin{split}
 Ex(\theta,\mm) =& 2|\mm| + \sum\limits_{ij}|v_i \times R_\theta v_j|\\
 =& 2|\mm| + \sum\limits_{i,j=1}^kl_il_j|\sin(\theta_i-\theta_j-\theta)|\\
 =& 2|\mm| +\sum\limits_{i,j=1}^kl_il_j\left(\frac{2}{\pi}-\frac{4}{\pi}\sum\limits_{m=1}^\infty\frac{\cos(2m(\theta_i-\theta_j-\theta)}{4m^2-1}\right)\\
 =& 2|\mm| + \frac{2}{\pi}\left(\sum\limits_{i=1}^kl_i\right)^2-\sum\limits_{i,j=1}^k\sum\limits_{m=1}^\infty \frac{4l_il_j}{\pi}\frac{\cos(2m\theta)\cos(2m(\theta_i-\theta_j))}{4m^2-1}\\
 =& 2|\mm|+ \frac{2}{\pi}\left(\frac{1}{2}\per(\mm)\right)^2-\sum\limits_{m=1}^\infty b^m(V)\cos(2m\theta).
\end{split}\end{equation}
From this we see the result holds. 
\end{proof}

Inspired by this formula, we propose the algorithm for reconstructing shapes from their excluded volumes. First we define a least-squares objective function to minimise. 

\begin{definition}
Let $M,k$ be input parameters. Let $L=(l_i)_{i=1}^k$, $\Theta=(\theta_i)_{i=1}^k$, without loss of generality taking $\theta_1=0$. Let $\hat{f}(m)=\frac{1}{2\pi}\int_{0}^{2\pi} f(\theta)\cos(m\theta)\,d\theta$. Define the objective function 
\begin{equation}\begin{split}
\mf(L,\Theta;M,k)
=& (b^0(V)-\langle f \rangle)^2 + \sum\limits_{m=1}^k(b^m(V)-\hat{f}(2m))^2,
\end{split}\end{equation}
with $b^m(V)$ defined as before in \Cref{propVolZono} in terms of $L,\Theta$, using $V = (l_iR_{\theta_i}e_1)_{i=1}^k$. 
\end{definition}

This objective function truly is a least-squares objective function in the $L^2$ sense too, as 
\begin{equation}\begin{split}
\mf(L,\Theta;M,k)= &||P_M(Ex(\cdot,\mm(V))-f)||_2^2,
\end{split}\end{equation}
where $P_M$ is the projection operator from $L^2$ onto the subspace spanned by the first $M$ Fourier nodes. 

Then we have an algorithm as an algebraic minimisation problem. While the system is highly nonlinear, it will be shown later that this is readily implementable.

\begin{algorithm}Given input parameters $M,k\in \mathbb{N}$, $f\in L^2$, perform the following. 
\begin{enumerate} 
\item Find the first $M$ even Fourier coefficients of $f$. 
\item Minimise $\mf$ with respect to $L,\Theta$. 
\item Reclaim $V$ from $L,\Theta$, and construct $\mm=Z(V)$.
\end{enumerate}
\end{algorithm}

\subsection{Estimates and convergence}

We have set up a least-squares optimisation problem for the reconstruction problem, the next step is to ensure that it produces meaningful solutions. Before obtaining results on the convergence of the scheme itself, some preliminary compactness and continuity type results are required. 

\subsubsection{Continuity of $Ex$ with respect to $\mm$}

\begin{proposition}
Assume $\mm_i$ are convex bodies and $\mm_i\to\mm$ in Hausdorff metric. Then $Ex(\cdot,\mm_i)\to Ex(\cdot,\mm)$ pointwise. 
\end{proposition}
\begin{proof}
As volume is continuous with respect to Hausdorff metric, it suffices to show that if $\mm_i\to \mm$, then $\mm_i -R\mm_i \to \mm-R\mm$ for each $R \in \mbox{SO}(3)$, however this is immediate as Minkowski addition and rotation are both continuous operations. 
\end{proof}

\begin{proposition}\label{propW1inftyBound}
For any convex body $\mm$, $\mbox{Lip}(Ex(\cdot,\mm))\leq \frac{1}{2}\per(\mm)^2$. 
\end{proposition}
\begin{proof}
First consider the case that $\mm\in C^2_+$. Let $h=h(\cdot,\mm)$ and $Ex=Ex(\cdot,\mm)$ for brevity. Recall we can write the excluded area function as a convolution-type object, 
\begin{equation}\begin{split}
&Ex(\theta-\pi)\\
=&2|\mm|+\frac{1}{2}\int_{0}^{2\pi} \big(h(\omega-\theta)+h(\omega+\theta)\big)\big(h(\omega)+h''(\omega)\big)\,d\omega.
\end{split}\end{equation}
This then implies, recalling $h+h''>0$,

\begin{equation}\begin{split}
|\nabla Ex(\theta-\pi,\mm)|=& \frac{1}{2}\left|\int_{0}^{2\pi} \big(-h'(\omega-\theta)+h'(\omega+\theta)\big)\big(h(\omega)+h''(\omega)\big)\,d\omega\right|\\
\leq & \int_{0}^{2\pi} ||h'||_\infty\big|h(\omega)+h''(\omega)\big|\,d\omega\\
=& ||h'||_\infty\int_{0}^{2\pi}h(\omega)+h''(\omega)\,d\omega\\
=& ||h'||_\infty \int_{0}^{2\pi} h(\omega)\,d\omega\\
=& \mbox{Lip}(h)\mbox{Per}(\mm)\leq \frac{1}{2}\mbox{Per}(\mm)^2. 
\end{split}\end{equation}
Note in the last line we used that $\mbox{Lip}(h)\leq \mbox{Diam}(\mm)\leq \frac{1}{2}\mbox{Per}(\mm)$. 

Now by density we prove this holds for all convex bodies. If $\mm_i\to\mm$, and $\mm_i\in C^2_+$ then the Lipschitz constants of $Ex(\cdot,\mm_i)$ must be bounded as $\mbox{Per}(\mm_i)\to\mbox{Per}(\mm)$. Furthermore as $Ex(\cdot,\mm_i)$ converges pointwise, by Arzela-Ascoli, $Ex(\cdot,\mm)$ must be Lipschitz with the corresponding bound. 
\end{proof}

\begin{proposition}
Let $h$ be the support function for a body $\mm\in C^2_+$ with $h+h''\geq \gamma$ for some $\gamma>0$. Then $Ex$ satisfies the differential inequality,
\begin{equation}\begin{split}
Ex(\theta,\mm)+\frac{d^2}{d\theta^2}Ex(\theta,\mm)\geq 2|\mm|+2\gamma\mbox{Per}\mm
\end{split}\end{equation}
in the classical sense.
\end{proposition}
\begin{proof}
Differentiating the excluded area function twice with respect to $\theta$ and adding $Ex(\theta,\mm)$ we obtain 
\begin{equation}\begin{split}
&Ex(\theta,\mm)+\frac{d^2}{d\theta^2}Ex(\theta,\mm)-2|\mm|\\
=& \int_0^{2\pi} \left(h''(\omega+\pi-\theta)+h(\omega+\pi-\theta)+ h(\omega+\pi+\theta)+h''(\omega+\pi+\theta)\right)\left(h(\omega)+h''(\omega)\right)\,d\omega\\
\geq & \int_0^{2\pi}(2\gamma)\left(h(\theta)+h''(\theta)\right)\,d\theta \\
=& 2\gamma\int_0^{2\pi}h(\theta)\,d\theta=2\gamma\mbox{Per}\mm.
\end{split}\end{equation}
\end{proof}

\begin{proposition}
For any $\mm\in C^2_+$, 
\begin{equation}\begin{split}
&\int_{\So}\left|\frac{d^2}{d\theta^2}Ex(\theta,\mm)\right|\,d\theta\leq 2\pi \langle Ex(\cdot,\mm)\rangle + 4\pi|\mm|+2\mbox{Per}(\mm)^2.
\end{split}\end{equation}
\end{proposition}
\begin{proof}
Recall that
\begin{equation}\begin{split}
  &Ex(\theta,\mm)+\frac{d^2}{d\theta^2}Ex(\theta,\mm)-2|\mm|\\
 =& \int_0^{2\pi} \big(h''(\omega+\pi-\theta)+h(\omega+\pi-\theta)+h(\omega+\pi+\theta)+h''(\omega+\pi+\theta)\big)\big(h(\omega)+h''(\omega)\big)\,d\omega.
\end{split}\end{equation}
This gives 
\begin{equation}\begin{split}\label{eq2ndDerivBound}
 &\left|\frac{d^2}{d\theta^2}Ex(\theta,\mm)\right|\\
 \leq & Ex(\theta,\mm)+2|\mm|+\int_0^{2\pi} \left(h''(\omega+\pi-\theta)+h(\omega+\pi-\theta)+h(\omega+\pi+\theta)+h''(\omega+\pi+\theta)\right)\left(h(\omega)+h''(\omega)\right)\,d\omega
\end{split}\end{equation}
In obtaining this inequality, the non-negativity of $h+h''$ was essential. 
Now we note that as $h$ is $C^2$ and $2\pi$-periodic, \begin{equation}\begin{split}
\int_0^{2\pi}h''(\omega)+h(\omega)\,d\omega=\int_0^{2\pi}h(\omega)\,d\omega=\mbox{Per}(\mm).
\end{split}\end{equation}
Thus integrating \ref{eq2ndDerivBound} with respect to $\theta$, we have 
\begin{equation}\begin{split}
&\int_0^{2\pi}\left|\frac{d^2}{d\theta^2}Ex(\theta,\mm)\right|\,d\theta\\
\leq &\int_0^{2\pi}Ex(\theta,\mm)\,d\theta +4\pi|\mm|+2\mbox{Per}(\mm)\int_0^{2\pi}h(\omega)\, d\omega\\
=& 2\pi \langle Ex(\cdot,\mm)\rangle + 4\pi|\mm|+2\mbox{Per}(\mm)^2.
\end{split}\end{equation}
Now by density, using that volume and perimeter are continuous with respect to Hausdorff metric and that $Ex$ is continuous in $L^\infty$ with respect to Hausdorff convergence, and $\frac{\partial}{\partial\theta}Ex(\cdot,\mm)$ can be controlled in $L^\infty$ norm as in \Cref{propW1inftyBound}, this implies that for any $\mm \in \mathcal{K}^2$, $\frac{d^2}{d\theta^2}Ex(\cdot,\mm)$ defines a Radon measure so that
\begin{equation}\begin{split}
&\left|\frac{d^2}{d\theta^2}Ex(\cdot,\mm)\right|(\So)\leq 2\pi \langle Ex(\cdot,\mm)\rangle + 4\pi|\mm|+2\mbox{Per}(\mm)^2.
\end{split}\end{equation}
\end{proof}

\begin{theorem}
Assume that $\mm$ is a convex body. Then $Ex(\cdot,\mm)$ is Lipschitz with $||\nabla Ex(\cdot,\mm)||_\infty <2\per(\mm)^2$, and $\nabla^2 Ex$ is a Radon measure with 
\begin{equation}\begin{split}
&\left|\frac{d^2}{d\theta^2}Ex(\theta,\mm)\right|(\So)\leq 2\pi \langle Ex(\cdot,\mm)\rangle + 4\pi|\mm|+2\mbox{Per}(\mm)^2.
\end{split}\end{equation}
Furthermore, if $\mm_i\to\mm$ in Hausdorff metric, then 
\begin{itemize}
\item $Ex(\cdot,\mm_i)\to Ex(\cdot,\mm)$ in $W^{1,p}$ for $p<\infty$. 
\item $Ex(\cdot,\mm_i)\overset{*}{\rightharpoonup} Ex(\cdot,\mm)$ in $W^{1,\infty}$.
\item $\nabla Ex(\cdot,\mm_i)\rightharpoonup \nabla Ex(\cdot,\mm)$ in $BV$. 
\end{itemize}
\end{theorem}
\begin{proof}
If $\mm$ is in $C^2_+$ the estimate has been shown previously. To extend to general convex bodies, we proceed by a density argument. Let $\mm$ be a convex body and $\mm_i\to \mm$ in Hausdorff metric and $\mm_i\in C^2_+$ for each $i$. As $\mm_i \to \mm$, $\per(\mm_i)\to\per(\mm)$, $|\mm_i|\to|\mm|$ and $Ex(\cdot,\mm_i)\to Ex(\cdot,\mm)$ pointwise. Since the Lipschitz norms of $Ex$ are bounded uniformly, and we have a unique pointwise limit for all subsequences, this implies that $Ex(\cdot,\mm_i)\overset{*}{\rightharpoonup} Ex(\cdot,\mm)$ in $W^{1,\infty}$. Furthermore, an application of Helly's selection theorem to $\nabla Ex(\cdot,\mm_i)$ we have that $Ex(\cdot,\mm_i)\to Ex(\cdot,\mm)$ in $W^{1,p}$ for $p<\infty$, and $\nabla^2 Ex(\cdot,\mm)$ is a Radon measure with the appropriate bound. 

To show the convergence result, the same reasoning as the density result is applied, except now $\mm_i$ needn't be $C^2_+$ bodies, as we know they satisfy the same bounds. 
\end{proof}

\begin{corollary}\label{corollarySteiner}
Let $\mm$ be a zonoid. Then 
\begin{equation}\begin{split}
\frac{1}{2\pi}\int_{0}^{2\pi}Ex(\theta,\mm)\,d\theta=2|\mm|+\frac{1}{2\pi}\per(\mm)^2.
\end{split}\end{equation}
\begin{proof}
If $\mm$ is a zonotope, this follows from integrating the result of \Cref{propVolZono}. By density, as $Ex(\cdot,\mm)$ using the previous result we can then extend this to all zonoids. 
\end{proof}
\end{corollary}

\subsubsection{Convergence of solutions}

\begin{definition}
Let $f\in L^2(0,2\pi)$. Then let $C(k)= \inf\limits_{|V|=k} ||Ex(\cdot,Z(V))-f||_2^2$, and ${C(\infty)=\inf\limits_{\mm} ||Ex(\cdot,\mm)-f||_2^2}$ where the infimum is taken over all zonoids $\mm$.
\end{definition}

\begin{proposition}
For all $k\in\mathbb{N}$ and $k=\infty$, the minimisation problem defining $C(k)$ admits a minimiser. 
\end{proposition}
\begin{proof}
For finite $k$, we use that the average of $Ex(\cdot,Z(V_i))$ must be bounded by \Cref{corollarySteiner} if $ ||Ex(\cdot,Z(V_i))-f||_2^2$ is bounded. This implies that the perimeters and thus diameters of $Z(V_i)$ must be bounded, and thus as each zonotope contains the origin, this implies that the zonotopes admit a Hausdorff metric converging subsequence. Furthermore, we can also take each $v_{i,j}=l_{i,j}R_{\theta_{i,j}}$ for $V_i=(v_{i,j})_{j=1}^k$ to have a converging subsequence as the perimeter is controlled by the sum of their norms. Then, as this implies uniform convergence of $Ex(\cdot,Z(V_i))\to Ex(\cdot,Z(V))$, we have that $||Ex(\cdot,Z(V_i))-f||_2^2\to ||Ex(\cdot,Z(V))-f||_2^2$. Thus by a direct method argument a minimiser exists \protect\cite{dacorogna2007direct}. 

For the case where $k=\infty$ the proof is nearly identical, except now we consider only $\mm_i$ and not its basis $V_i$. We control the perimeter, and thus diameter, by the average of $Ex$, and thus obtain Hausdorff compactness by the same argument, at which point the result follows by the same reasoning. 
\end{proof} 

\begin{proposition}
$C(k)$ is a decreasing in $k$, and $\lim\limits_{k\to\infty}C(k)=C(\infty)$.
\end{proposition}
\begin{proof}
To see it is decreasing in $k$ it suffices to observe that if $V$ is a list of $k$ vectors, then $V'=(v_1,v_2,...,v_k,0)$ is a list of $k+1$ vectors, and $Z(V)=Z(V')$, so in essence $V$ is a candidate minimiser for $\min\limits_{|U|=k+1}||Ex(\cdot,Z(U))-f||_2^2$. As $C(k)$ is a decreasing sequence with a lower bound (zero), this means it attains a limit. To see that the limit is $C(\infty)$, let $\mm$ be a minimiser of the minimisation problem defining $C(\infty)$. As $\mm$ is a zonoid, by definition there exists a sequence of zonotopes $Z(V_k)$ with $Z(V_k)\to \mm$, $|V_k|=k$. Then $Ex(\cdot,Z(V_k))\to Ex(\cdot,\mm)$ uniformly, so $C(k)\leq||Ex(\cdot,Z(V_k))-f||_2^2$. Taking the limit as $k\to\infty$ of both sides gives 
\begin{equation}\begin{split}
\lim\limits_{k \to \infty }C(k)\leq ||Ex(\cdot,\mm)-f||_2^2=C(\infty)\leq \lim\limits_{k\to\infty}C(k).
\end{split}\end{equation}
\end{proof}

\begin{lemma}\label{lemmaExFourierDecay}
Let $M\in\mathbb{N}$ be greater than $1$, and $\mm$ be a zonoid. Then there exists a constant $C>0$, independent of $M$ and $\mm$, so that 
\begin{equation}\begin{split}
||(I-P_M)Ex(\cdot,\mm)||_2^2\leq \frac{C}{M^3}\per(\mm)^4.
\end{split}\end{equation} 
\end{lemma}
\begin{proof}
First we prove the result when $\mm$ is a zonotope. Let $\mm=Z(V)$ for a set of spanning vectors $V=(v_i)_{i=1}^k$. Define $l_i,\theta_i$ to be the norm and angles of $v_i$ respectively. We can then recall the Fourier decomposition of the excluded area function, and obtain a bound on the $L^2$ norm of the projection as 
\begin{equation}\begin{split}
&||(I-P_M)Ex(\cdot,Z(U))||_2^2\\
=& \sum\limits_{m=M+1}^\infty |b^m(V)|^2\\
=& \sum\limits_{m=M+1}^\infty \left|\sum\limits_{ij}\frac{4l_il_j\cos(2m(\theta_i-\theta_j))}{\pi(4m^2-1)}\right|^2\\
\leq & \sum\limits_{m=M+1}^\infty \left(\sum\limits_{ij}\frac{4l_il_j}{\pi(4m^2-1)}\right)^2\\
=&\sum\limits_{m=M+1}^\infty\frac{16}{\pi(4m^2-1)^2} \left(\sum\limits_{i}l_i\right)^4\\
=&\sum\limits_{m=M+1}^\infty\frac{16}{\pi(4m^2-1)^2} \left(\frac{1}{2}\per(Z(U))\right)^4\\
\leq & C_1\sum\limits_{m=M+1}^\infty \frac{\per(Z(U))^4}{m^4}\\
\leq & C \frac{\per(Z(U))^4}{M^3}
\end{split}\end{equation}
for appropriate positive constants $C_1,C$ which are independent of $M$ and $Z(U)$. Then by density of zonotopes in the zonoids with respect to Hausdorff metric, we obtain the result for zonoids using the continuity properties of the excluded area map and perimeter. 
\end{proof}

\begin{theorem}\label{theoremConverge}
Let $V_{M,k}$ denote any solutions to to the minimisation problem $\min\limits_{L,\Theta}\mf(L,\Theta;M,k)$. Then for any choice of $M_i,k_i$ so that $M_i\to\infty$, $k_i\to\infty$, we have that there exists a subsequence $M_{ij},k_{ij}$ and zonoid $\mm$ so that 
\begin{equation}\begin{split}
Ex(\cdot,Z(V_{M_{ij},k_{ij}}))\overset{*}{\rightharpoonup}&Ex(\cdot,\mm) \hspace{1cm} (W^{1,\infty}),\\
Ex(\cdot,Z(V_{M_{ij},k_{ij}}))\to &Ex(\cdot,\mm) \hspace{1cm} (W^{1,p})\, (p<\infty),\\
\frac{d^2}{d\theta^2}Ex(\cdot,Z(V_{M_{ij},k_{ij}}))\overset{*}{\rightharpoonup} &\frac{d^2}{d\theta^2}Ex(\cdot,\mm) \hspace{0.4cm} C(\mathbb{S}^1)^*,\\
Z(V_{M_{ij},k_{ij}})\to &\mm \hspace{2cm} (\mbox{Hausdorff}),\\
||Ex(\cdot,\mm)-f||_2^2=&\inf\limits_{\tilde\mm} ||Ex(\cdot,\tilde\mm)-f||_2^2=C(\infty).
\end{split}\end{equation}
In particular, if $f$ can be written as an excluded area function, then $f=Ex(\cdot,\mm)$ also. 
\end{theorem}
\begin{proof}
Let $V_i$ denote the minimisers, and $L_i,\Theta_i$ denote the lengths and angles (respectively) of the vectors of $V_i$. Then 
\begin{equation}\begin{split}
C(k_i)=& \min\limits_{|U|=k_i}||Ex(\cdot,Z(U))-f||_2^2\\
\geq &  \min\limits_{|U|=k_i}||P_{M_i}(Ex(\cdot,Z(U))-f)||_2^2\\
=& \mf(L,\Theta;M_i,k_i)\\
\geq &A_1\left(\overline{Ex(\cdot,Z(V_i))}^2-1\right)
\end{split}\end{equation}
for an appropriate positive constant $A_1>0$, which is indepedent of $k$. This means that for our sequence of minimisers, the average of $Ex(\cdot,Z(V_i))$ is bounded, and thus the perimeters of $Z(V_i)$ admit uniform control. This implies immediately that we can take a subsequence $V_{ij}$ so that $Z(V_{ij})\to\mm$ for some zonoid $\mm$ and $Ex(\cdot,\mm_i)\overset{*}{\rightharpoonup} Ex(\cdot,\mm)$. It suffices to show that $\mm$ is a minimiser of $\min\limits_{\tilde{\mm}}||Ex(\cdot,\tilde{\mm}))-f||_2^2$. To see this, we note that 
\begin{equation}\begin{split}
&||P_M(Ex(\cdot,Z(U))-f)||_2^2\\
= & ||Ex(\cdot,Z(U))-f||_2^2-||(I-P_M)(Ex(\cdot,Z(U))-f)||_2^2\\
\geq &C(k)-2\left(||(I-P_M)Ex(\cdot,Z(U))||_2^2+||(I-P_M)f||_2^2\right).
\end{split}\end{equation}
Therefore if $U=V_{ij}$, which has bounded perimeter, by \Cref{lemmaExFourierDecay} this means that 
\begin{equation}\begin{split}
&||P_{M_{ij}}(Ex(\cdot,Z(V_{ij}))-f)||_2^2\\
\geq & C(k_{ij})-\frac{A_2}{M_{ij}^3}-2||(I-P_{M_{ij}})f||_2^2,
\end{split}\end{equation}
where $A_2$ is a positive constant controlled by the perimeter of $V_{ij}$, and hence uniformly controlled for all $i,j$. As $f \in L^2$, we have that $||(I-P_{M_{ij}})f||_2^2\to 0$ as $j\to \infty$, and therefore 
\begin{equation}\begin{split}
&\liminf\limits_{j \to \infty} ||P_{M_{ij}}(Ex(\cdot,Z(V_{ij}))-f)||_2^2 \\
\geq &\lim\limits_{j\to\infty}C(k_{ij})=C(\infty).
\end{split}\end{equation}
However we have also seen we have an estimate from above, that 
\begin{equation}\begin{split}
||P_{M_{ij}}(Ex(\cdot,Z(V_{ij}))-f)||_2^2\leq C(k_{ij})
\end{split}\end{equation}
which similarly implies $\limsup\limits_{j\to\infty}||P_{M_{ij}}(Ex(\cdot,Z(V_{ij}))-f)||_2^2\leq C(\infty)$. This means that $\lim\limits_{j \to \infty}||P_{M_{ij}}(Ex(\cdot,Z(V_{ij}))-f)||_2^2= C(\infty)$, but since \begin{equation}|P_{M_{ij}}(Ex(\cdot,Z(V_{ij}))-f)||_2^2=||(Ex(\cdot,Z(V_{ij}))-f)||_2^2+o(1),\end{equation} this implies that $||Ex(\cdot,Z(V_{ij}))- f||_2^2\to C(\infty)$. Since $Ex(\cdot,Z(V_{ij}))$ converges weak-* in $W^{1,\infty}$ to $Ex(\cdot,\mm)$, this implies that $Ex(\cdot,Z(V_{ij}))$ converges weak-* in $W^{1,\infty}$ to a best $L^2$ approximation of $f$ in the space of excluded area functions of zonotopes. Furthermore, if $C(\infty)=0$, this implies that $Ex(\cdot,Z(V_{ij})\to f$ in $L^2$, so that $Ex(\cdot,\mm)=f$. 

\end{proof}

\subsection{Convergence rate}

We now turn to obtaining a convergence rate, assuming that a solution exists. In this case we will first need some estimates on approximations of zonoids in the Hausdorff sense.

\begin{definition}
Let $\mm$ be a zonoid with support function $h$. let $u_\theta$ denote the unit vector at angle $\theta$. Define the $k$-th canonical approximation to $\mm$ to be given by 
\begin{equation}\begin{split}
\bigcap\limits_{i=1}^{4k}\left\{ x \in \mathbb{R}^2 : u_{\frac{i\pi}{2k}}\cdot x < h\left(\frac{i\pi}{2k}\right) \right\}
\end{split}\end{equation}
\end{definition}

\begin{remark}
If $\mm_k$ is the $k$-th canonical approximation to $\mm$, and $h_k$ denotes its support function, then we have immediately that $h_k\left(\frac{i\pi}{2k}\right)=h\left(\frac{i\pi}{2k}\right)$ for $i=1,...,4k$, and that ${\mm\subset {\mm}_k}$. Furthermore, we see that $\mm_k$ is a centrally symmetric polygon with at most $4k$ edges, and thus is a zonotope with at most $2k$ spanning vectors. 
\end{remark}

\begin{lemma}
Let $\mm$ be a zonoid and its $k$-th canonical approximation be $\mm_k$. Let $d$ denote the diameter of $\mm$, and $d_k$ be the diameter of $\mm_k$. Then for $k>1$,
\begin{equation}\begin{split}
d_{k}\leq & d\sqrt{2},\\
\mbox{Per}({\mm}_{k})\leq &  4d.
\end{split}\end{equation}
\end{lemma}
\begin{proof}
$\mm_k\subset \mm_1$ for all $k\geq 1$. This means that $\mm_{2k}$ is contained inside the rectangle $\left[-h(\pi),h(0)\right]\times\left[-h\left(\frac{3\pi}{2}\right),h\left(\frac{\pi}{2}\right)\right]$. Therefore it is contained inside a square with side lengths less than $2||h||_\infty$. We then recall that $d=2||h||_\infty$, so that the diameter of this box is $\sqrt{2}d$. Therefore, as diameter is monotone, this implies that $d_{2k}\leq \sqrt{2}d$.

We use a similar heuristic for the perimeter. Since perimeter is monotone for convex bodies, this means that $d_k$ is less than the perimeter of a square with side lengths $d$, which is $4d$. 
\end{proof}

\begin{lemma}\label{lemmaCanonicalApprox}
Let $\mm$ be a zonoid and $\mm_k$ its $k$-th canonical approximation. Let $h$ denote the support function of $\mm$, $h_k$ the support function of $\mm_k$, and $d$ be the diameter of $\mm$. Then 
\begin{equation}\begin{split}
|h_k(\theta)-h(\theta)|\leq\frac{\pi(1+\sqrt{2})}{2k}d
\end{split}\end{equation}
\end{lemma}
\begin{proof}
First we recall that for a convex body centred at $0$, the Lipschitz constant of the support function is bounded by half the diameter. Thus $||h'||_\infty < d$, $ ||h_k'||_\infty \leq d_k \leq d\sqrt{2}$. Let $\theta_i=\frac{i\pi}{2k}$. Then we have that 
\begin{equation}\begin{split}
&||h_k-h||_\infty \\
=&\max\limits_{0\leq \theta \leq 2\pi}|h_k(\theta)-h(\theta)|\\
=&\max_{1\leq i \leq 4k}\max\limits_{\theta_{i}\leq \theta \leq \theta_{i+1}}|h_k(\theta)-h(\theta)|\\
=&\max_{1\leq i \leq 4k}\max\limits_{\theta_{i}\leq \theta \leq \theta_{i+1}}\left|h_k\left(\theta_{i}\right)-h(\theta_{i})+\int_{\theta_{i}}^\theta h_k'(t)-h'(t)\,dt\right|\\
=&\max_{1\leq i \leq 4k}\max\limits_{\theta_{i}\leq \theta \leq \theta_{i+1}}\left|\int_{\theta_{i}}^\theta h_k'(t)-h'(t)\,dt\right|\\
\leq &\max_{1\leq i \leq 4k}\max\limits_{\theta_{i}\leq \theta \leq \theta_{i+1}}|\theta_{i}-\theta_{i+1}|\big(||h'||_\infty +||h_k'||_\infty\big)\\
\leq & \frac{\pi(1+\sqrt{2})}{2k}d
\end{split}\end{equation}
\end{proof}

\begin{proposition}
Let $\mm$ be a zonoid with diameter $d$, and $\mm_k$ denote its $k$-th canonical approximation. Then there is a constant $C$ depending only on $d$ so that 
\begin{equation}\begin{split}
||Ex(\cdot,\mm)-Ex(\cdot,{\mm}_{2k})||_\infty\leq \frac{C}{k}.
\end{split}\end{equation}
\end{proposition}
\begin{proof}
Let $\epsilon =d_H(\mm,\mm_k)=||h_k-h||_\infty$, and $R\in \mbox{SO}(2)$. Let $B_\epsilon$ denote the ball of radius $\epsilon$ in $\mathbb{R}^2$. Then using results from \protect\cite[Section 4.1]{schneider2014convex} we may write
\begin{equation}\begin{split}
&|\mm-R\mm|\\
\leq &|({\mm}_{2k}+B_\epsilon)-R({\mm}_{2k}+B_\epsilon)|\\
=&|{\mm}_{2k}-R{\mm}_{2k}+2B_\epsilon|\\
=&|{\mm}_{2k}-R{\mm}_{2k}|+2\epsilon \mbox{Per}({\mm}_{2k}-R{\mm}_{2k})+4\pi\epsilon^2\\
=& Ex(R,{\mm}_{2k})+4\epsilon\mbox{Per}({\mm}_{2k})+4\pi\epsilon^2.
\end{split}\end{equation}
Applying the same inequality with $\mm,\mm_k$ interchanged gives 
\begin{equation}\begin{split}
&||Ex(\cdot,\mm)-Ex(\cdot,{\mm}_{2k})||_\infty \\
\leq &4\epsilon\max(\mbox{Per}(\mm),\mbox{Per}({\mm}_{2k}))+4\pi\epsilon^2\\
\leq & 4\epsilon\max(4d,4d)+4\pi\epsilon^2\\
\leq & 4\epsilon\max(\mbox{Per}(\mm),2\mbox{Per}(\mm))+4\pi\epsilon^2\\
=& 16\epsilon d+4\pi\epsilon^2.
\end{split}\end{equation}
Recalling \Cref{lemmaCanonicalApprox} gives the result. 
\end{proof}

\begin{theorem}\label{theoremConvergeRate}
Assume that $f=Ex(\cdot,\mm)$ for some zonoid $\mm$. Let $\mm_{M,k}$ denote a corresponding minimiser of $\mf$. Then 
\begin{equation}\begin{split}
&||Ex(\cdot,{\mm}_{M,k})-f||_2^2\\
&\leq \frac{C_1}{k^2}+\frac{C_2}{M^3}+2\sum\limits_{|n|\geq M}|\hat{f}(n)|^2.
\end{split}\end{equation}
\end{theorem}
\begin{proof}
Let $\mm_j$ denote the $j$-th canonical approximation of $\mm$. We recall the estimate $||(I-P_M)Ex(\cdot,{\mm}_{M,k})||_2^2\leq \frac{C_2}{M^3}$ obtained in  \Cref{lemmaExFourierDecay}. Then
\begin{equation}\begin{split}
&||Ex(\cdot,{\mm}_{M,k})-f||_2^2\\
=& ||P_M(Ex(\cdot,{\mm}_{M,k})-f)||_2^2+||(I-P_M)(Ex(\cdot,{\mm}_{M,k})-f)||_2^2\\
\leq & ||P_M(Ex(\cdot,{\mm}_j)-f)||_2^2+2||(I-P_M)Ex(\cdot,{\mm}_{M,k})||_2^2+2 ||(I-P_M)f||_2^2\\
\leq & ||Ex(\cdot,{\mm}_j)-f||_2^2+\frac{C_2}{M^3}+2\sum\limits_{|n|\geq M}|\hat{f}(n)|^2\\
\leq & 2\pi||Ex(\cdot,{\mm}_j)-f||_\infty^2+\frac{C_2}{M^3}+2\sum\limits_{|n|\geq M}|\hat{f}(n)|^2\\
\leq &\frac{C_1}{k^2}+\frac{C_2}{M^3}+2\sum\limits_{|n|\geq M}|\hat{f}(n)|^2.
\end{split}\end{equation}
\end{proof}

The decay rate given has significant consequences. It implies that i) Both $k$ and $M$ need to be large, and ii) that if $f$ has many high frequency contributions in Fourier space, the decay rate will suffer. In particular, if $f$ is less regular we expect a slower convergence. 

\subsection{Examples}\label{subsecExamples}
For some representative candidate functions $f$ and values $k,M$, we present constructed solutions to the minimisation problem, with an implementation in Wolfram Mathematica. For each case we include the zonotope produced and the excluded area graph for comparison with the candidate function. The candidate function $f$ is graphed in dashed orange, and the reconstructed excluded area function is given in blue. We also give the $L^2$ error of the result, $\epsilon$, to two significant digits.

It should be noted with regards to implementation that the minimisation problem on the basis vectors is highly non-linear, with generally non-unique solutions, either by symmetry or otherwise. Even more so, it is near impossible to rule out the possibility that the numerical schemes settle towards {\it local} minima of our objective function, rather than global minima. However, in these implementations these issues did not seem to be problematic.

The implementation however can be tackled using standard software, and in the following examples the results were obtained by an application of Wolfram Mathematica's inbuilt {\tt NMinimize} command, with the zonotope given in $(l_i,\theta_i)$ coordinates. To enforce the constraint $l_i\geq 0$, a variable $r_i^2=l_i$ is used which is unconstrained, and $\theta_1=0$ without loss of generality to remove the degeneracy of rotational symmetry. Otherwise it is is a very straightforward implementation, and the longest run time for the following applications was around 10 minutes on a desktop computer.

\begin{example}

First we consider a well behaved candidate excluded volume function, that appears to be ``well posed", of $f(\theta)=2+\frac{\sin(x)^2}{2}$. In fact we can prove this is the excluded area function of a shape with support function $h(\theta)=1+\frac{1}{2\sqrt{3}}\cos(2\theta)$, which results from a quick application of \Cref{lemEqVolume}. We note that as the support function has only one non-zero Fourier node, it is the unique shape, modulo symmetry, that produces the given excluded area function. The reconstructed zontopes and excluded area functions are given in \Cref{figSine}. Generally we see what we expect, the error is smaller with increasing $k,M$, and the solutions should be converging. As the support function has only one non-zero Fourier node after first order, we have that the reconstructed shapes must, modulo symmetry, converge to the correct convex domain. We can explicitly construct the body with support function $h$, which is included for comparison.

\begin{figure*}\begin{center}
\begin{subfigure}{0.49\textwidth}
\includegraphics[width=0.45\textwidth]{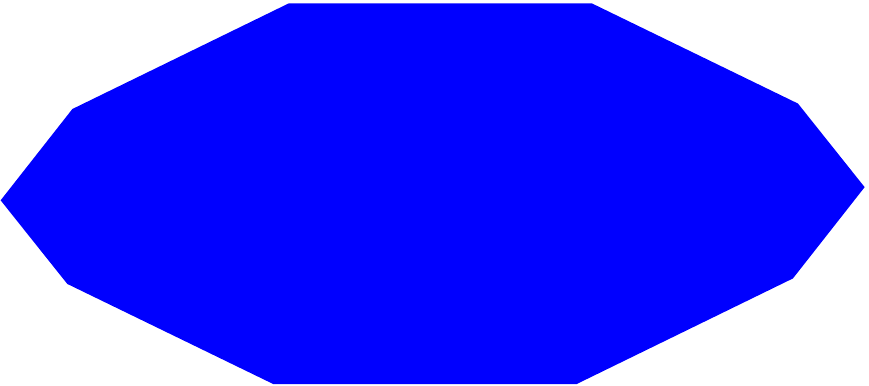}\includegraphics[width=0.45\textwidth]{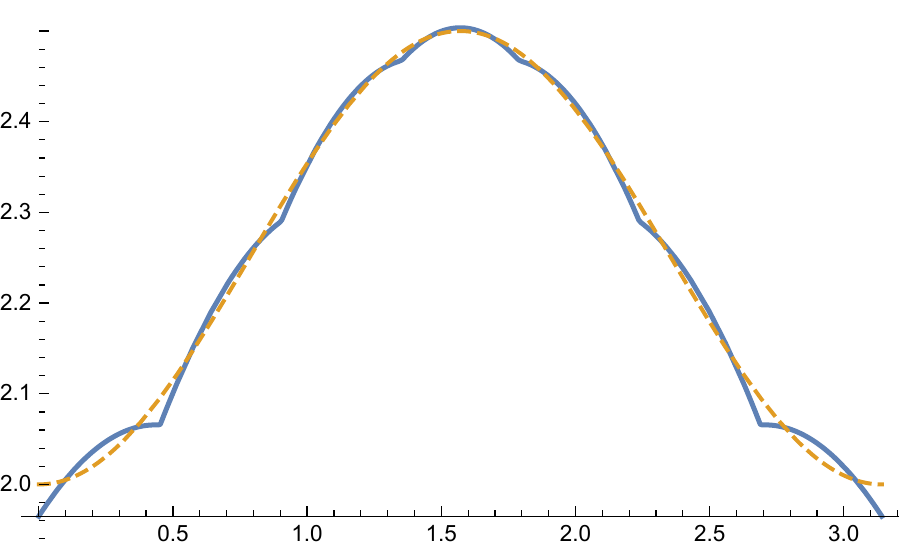}
\caption{$k=5,M=5,\epsilon =0.027$}\end{subfigure}
\begin{subfigure}{0.49\textwidth}
\includegraphics[width=0.45\textwidth]{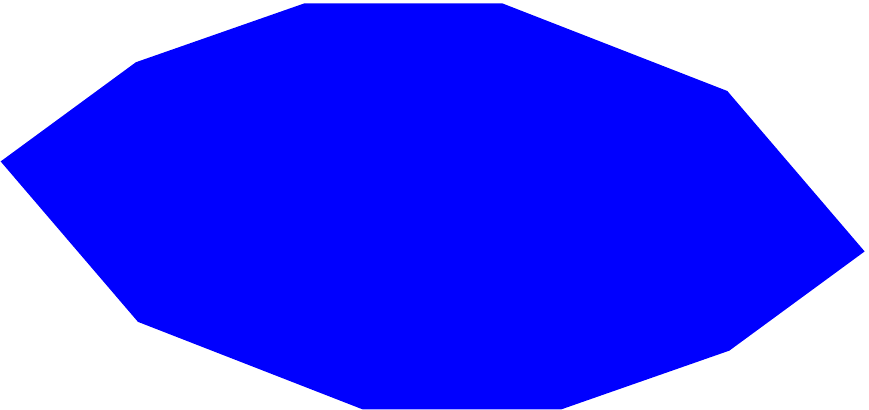}\includegraphics[width=0.45\textwidth]{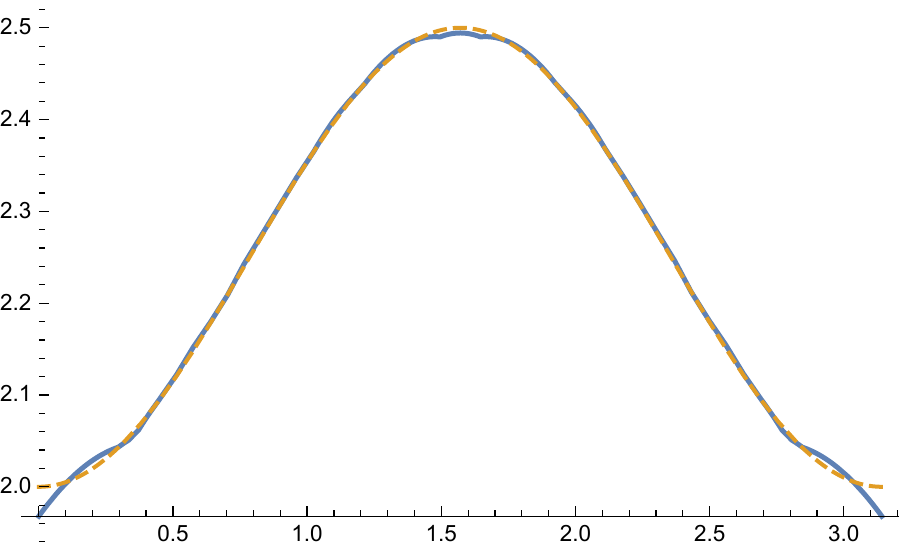}
\caption{$k=5,M=20,\epsilon=0.013$}
\end{subfigure}
\begin{subfigure}{0.49\textwidth}
\includegraphics[width=0.45\textwidth]{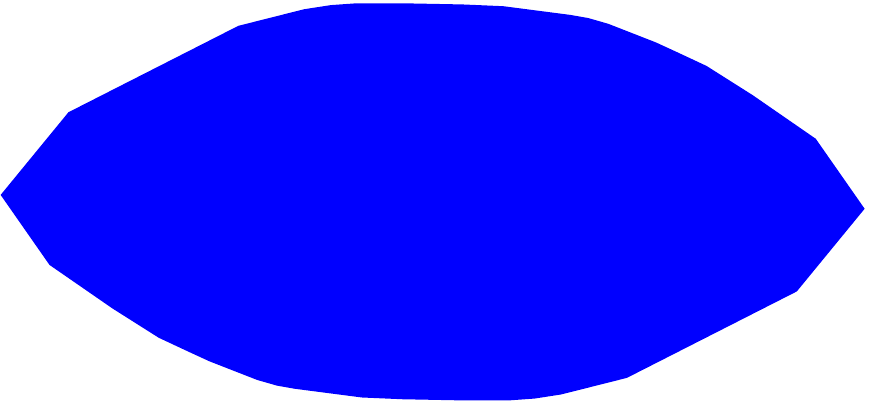}\includegraphics[width=0.45\textwidth]{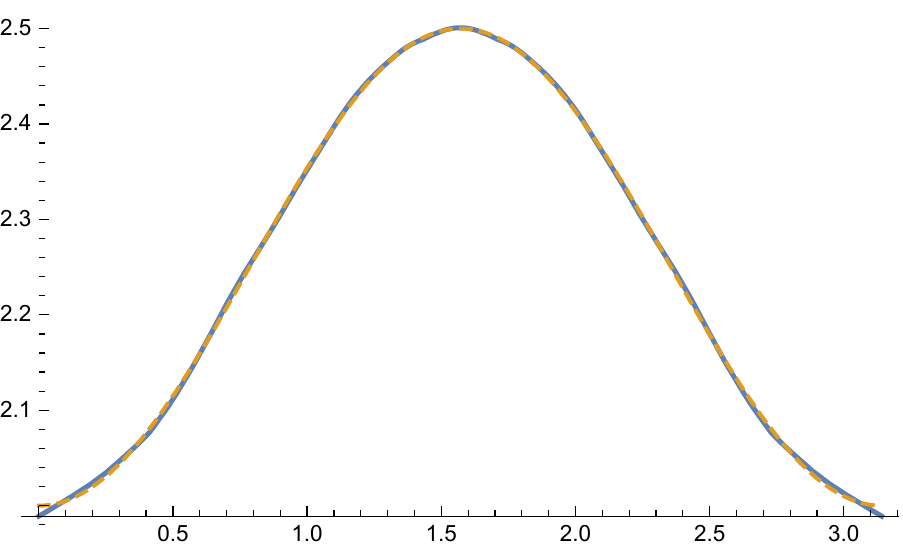}
\caption{$k=20,M=5,\epsilon=0.0065$}
\end{subfigure}
\begin{subfigure}{0.49\textwidth}
\includegraphics[width=0.45\textwidth]{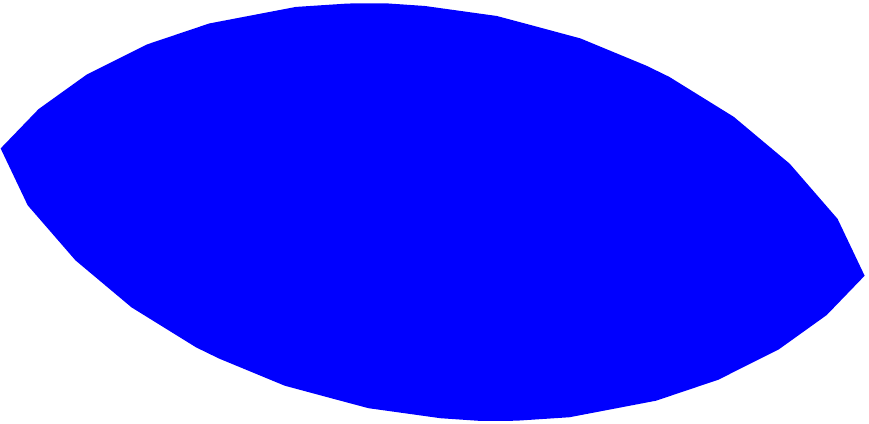}\includegraphics[width=0.45\textwidth]{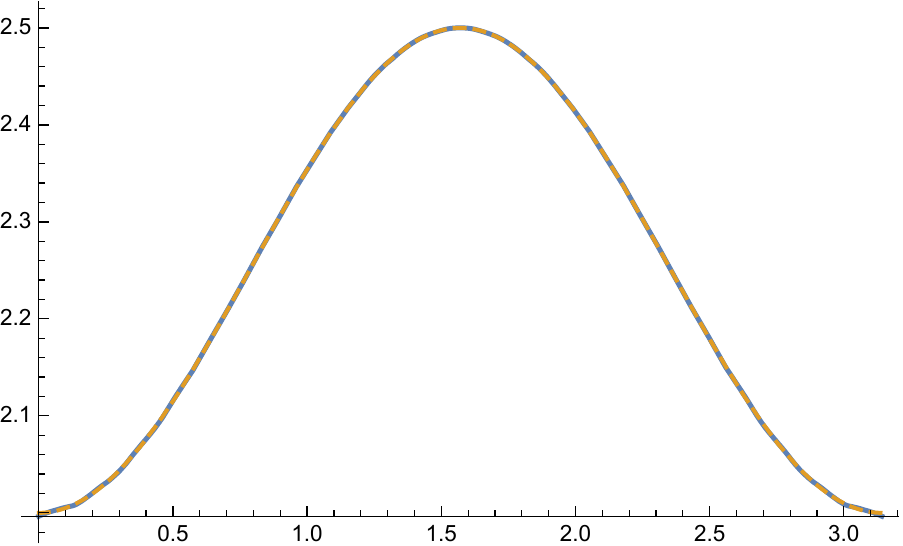}
\caption{$k=20,M=20,\epsilon=0.0014$}
\end{subfigure}
\begin{subfigure}{0.5\textwidth}\begin{center}
\includegraphics[width=0.5\textwidth]{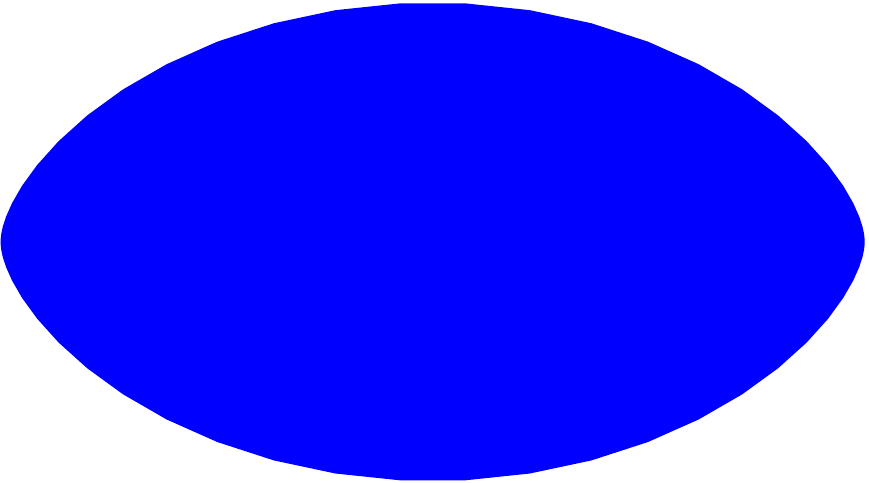}\end{center}
\caption{The true body $\mm$ so that $Ex(\cdot,\mm)=f$}
\end{subfigure}
\caption{Reconstructed bodies from $f(\theta)=2+\frac{\sin(\theta)^2}{2}$}\label{figSine}\end{center}
\end{figure*}

\end{example}

\begin{example}

Next we consider the function $f(\theta)=2+\frac{1}{2}|\sin(\theta)|-\frac{1}{2}\cos(2\theta)$. This is chosen as the function has infinitely many Fourier nodes and has no obvious convex body of which it is the support function. The reconstructed zonotopes and excluded area functions are included in \Cref{figJag}. In this case we appear to have convergence, but it seems apparent that increasing the number of Fourier nodes $M$ is more important than increasing the number of spanning vectors for the zonotope $k$. This is likely due to the high order Fourier nodes of $f$ causing errors for small $M$.

\begin{figure*}\begin{center}
\begin{subfigure}{0.49\textwidth}
\includegraphics[width=0.475\textwidth]{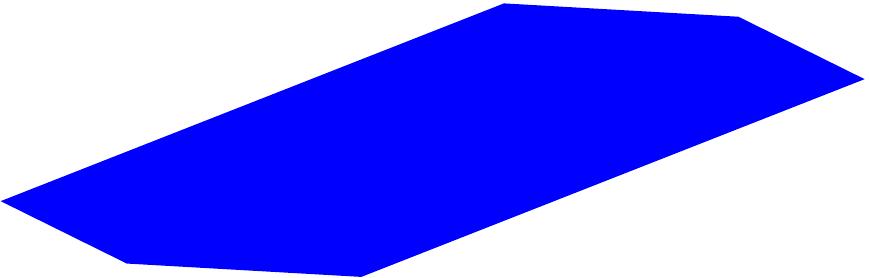}
\includegraphics[width=0.475\textwidth]{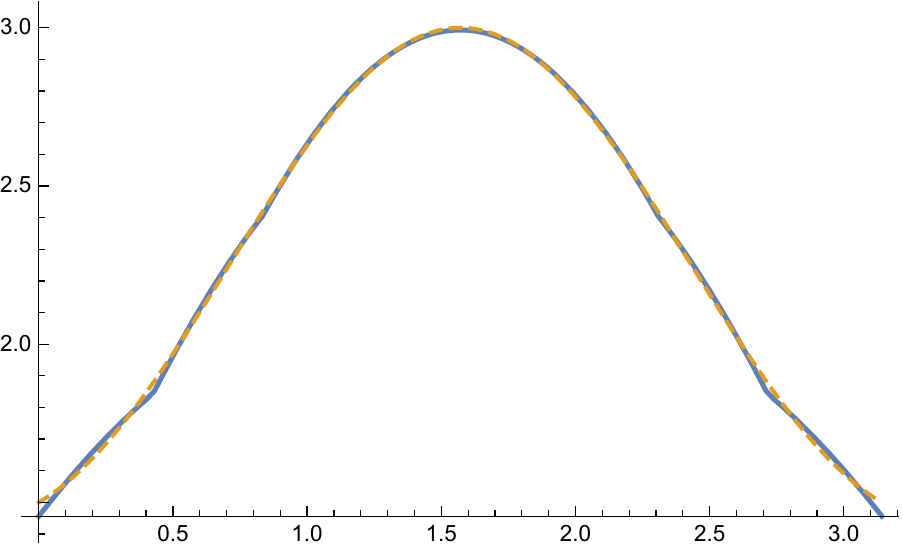}
\caption{$k=5,M=5,\epsilon =0.029$}
\end{subfigure}
\begin{subfigure}{0.49\textwidth}
\includegraphics[width=0.475\textwidth]{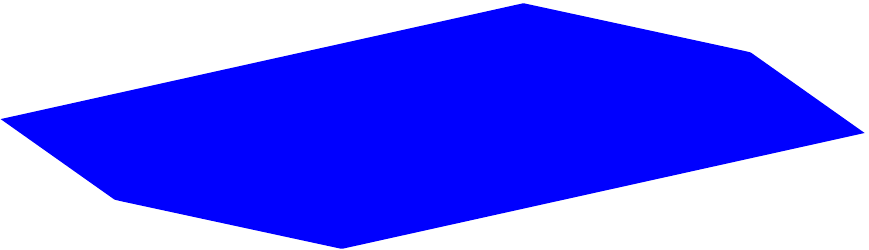}
\includegraphics[width=0.475\textwidth]{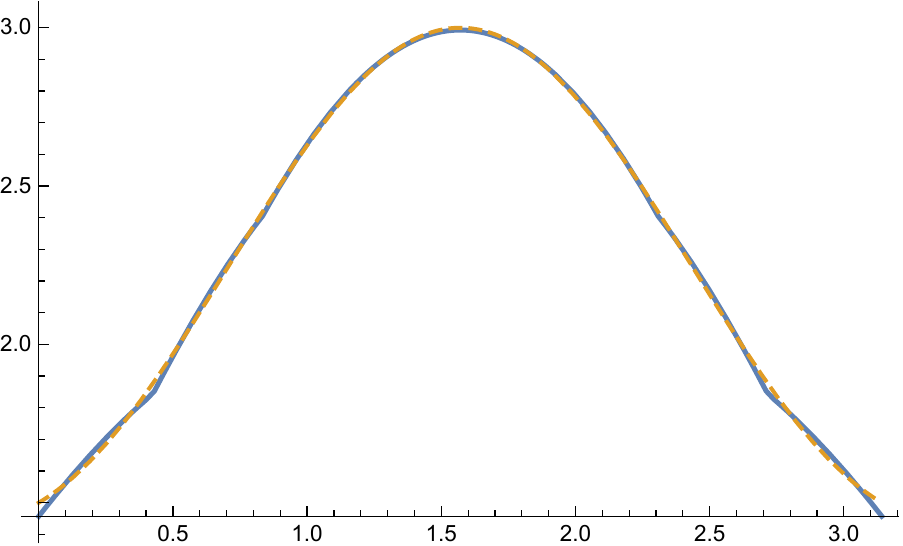}
\caption{$k=20,M=5,\epsilon =0.028$}
\end{subfigure}
\begin{subfigure}{0.49\textwidth}
\includegraphics[width=0.475\textwidth]{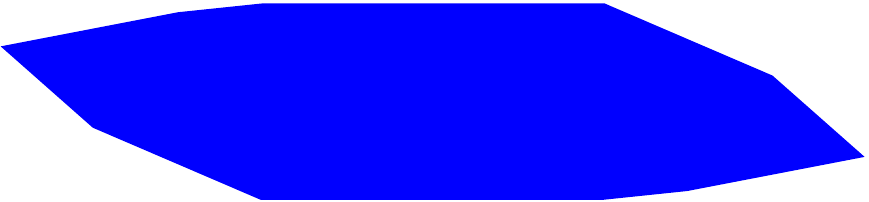}
\includegraphics[width=0.475\textwidth]{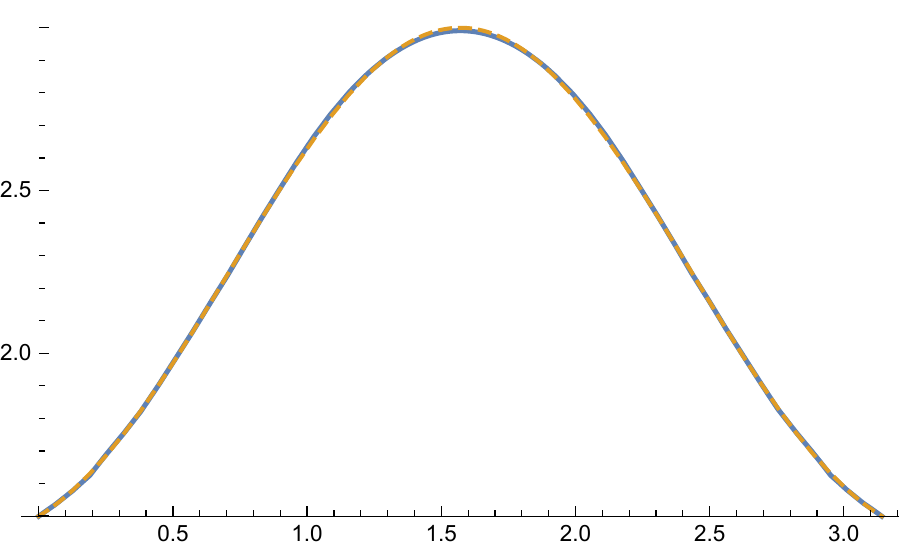}
\caption{$k=5,M=20,\epsilon =0.010$}
\end{subfigure}
\begin{subfigure}{0.49\textwidth}
\includegraphics[width=0.475\textwidth]{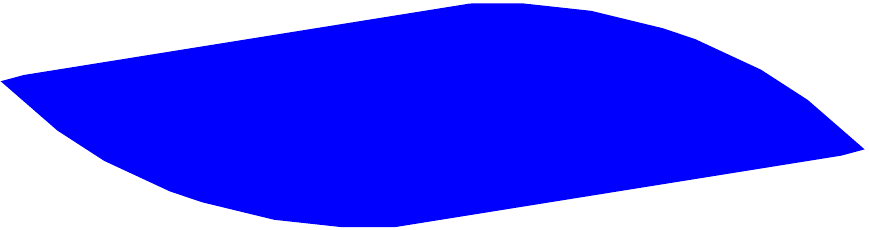}
\includegraphics[width=0.475\textwidth]{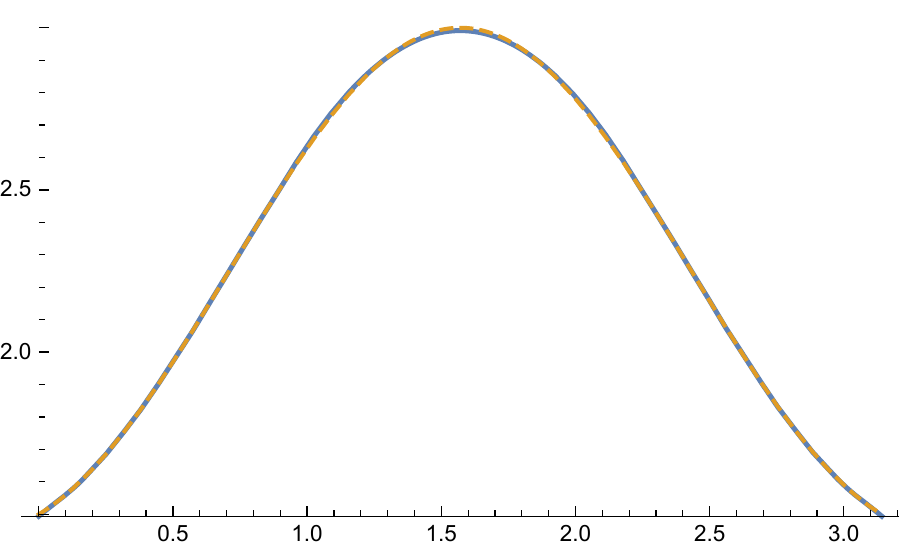}
\caption{$k=20,M=20,\epsilon =0.0097$}
\end{subfigure}
\caption{Reconstructed figures for $f(\theta)=2+\frac{|\sin(\theta)|-\cos(2\theta)}{2}$}\label{figJag}
\end{center}
\end{figure*}

\end{example}

\begin{example}
Consider $f(\theta)=1+\sin(\theta)^{30}$. This is chosen as it is not the excluded area function for any zonoid, which can be seen as follows. For this function numerically evaluate that its average is given by $\bar{f}\approx 1.14$, its Lipschitz constant is given by $\mbox{Lip}(f)\approx 3.35$, and $f(0)=1$. So we have that $\bar{f}-\frac{1}{2}f(0)\approx 0.64$. If $f$ were an excluded area function for a zonoid, it would have to satisfy $\bar{f}-\frac{1}{2}f(0)\geq \frac{1}{\pi}\mbox{Lip}(f)$. This can be seen by integrating \Cref{propVolZono}, which shows 
\begin{equation}\begin{split}&\frac{1}{2\pi}\int_0^{2\pi}Ex(\theta,\mm)\,d\theta\\
=&2|\mm|+\frac{1}{2\pi}\mbox{Per}(\mm)^2\\
=&\frac{1}{2}Ex(0)+\frac{1}{2\pi}\mbox{Per}(\mm)^2,\end{split}\end{equation}
and comparing this to the estimate in \Cref{propW1inftyBound} for the perimeter. However, numerically we evaluate $0.64\approx\bar{f}-\frac{1}{2}f(0)$ and $\frac{1}{\pi}\mbox{Lip}(f)\approx 1.1$, Thus we see this is not an excluded area function, but perform the analysis to see how the solutions to the algorithm appear. The reconstructed zonotopes and excluded area functions are in \Cref{figDirac}. As expected, we do not appear to have any convergence of our solutions to $f$. In this case we see that for $k=5,20$ and $M=5,20$, there is very little difference in behaviour of solutions, certainly nothing to the naked eye. However, it is expected that the reconstructed shape is a best approximation in $L^2$ to $f$ within all possible excluded area functions. 

\begin{figure*}
\begin{center}
\begin{subfigure}{0.49\textwidth}
\includegraphics[width=0.45\textwidth]{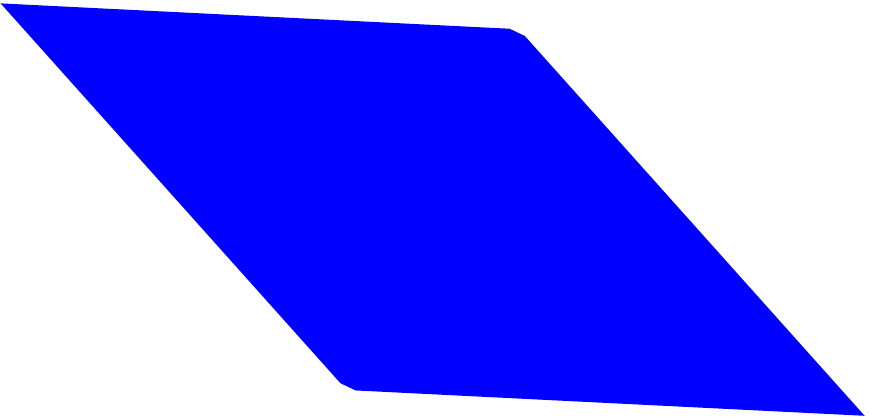}
\includegraphics[width=0.45\textwidth]{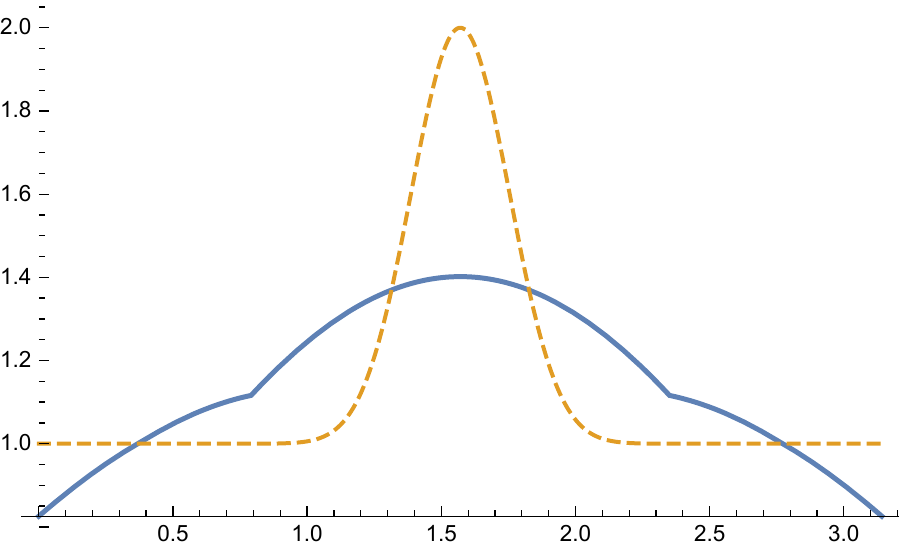}
\caption{$k=5,M=5,\epsilon =0.53$}
\end{subfigure}
\begin{subfigure}{0.49\textwidth}
\includegraphics[width=0.45\textwidth]{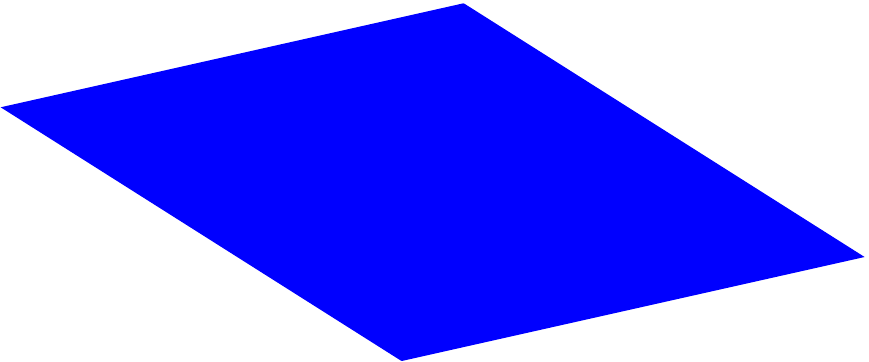}
\includegraphics[width=0.45\textwidth]{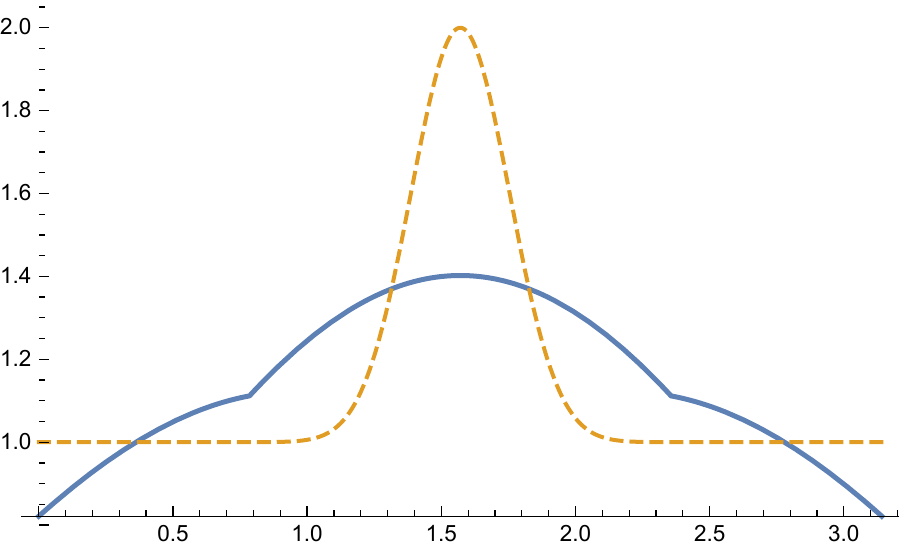}
\caption{$k=20,M=5,\epsilon =0.53$}
\end{subfigure}
\begin{subfigure}{0.49\textwidth}
\includegraphics[width=0.45\textwidth]{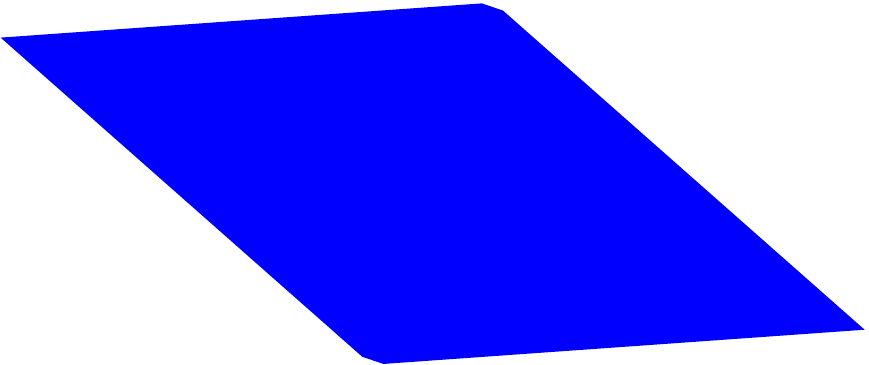}
\includegraphics[width=0.45\textwidth]{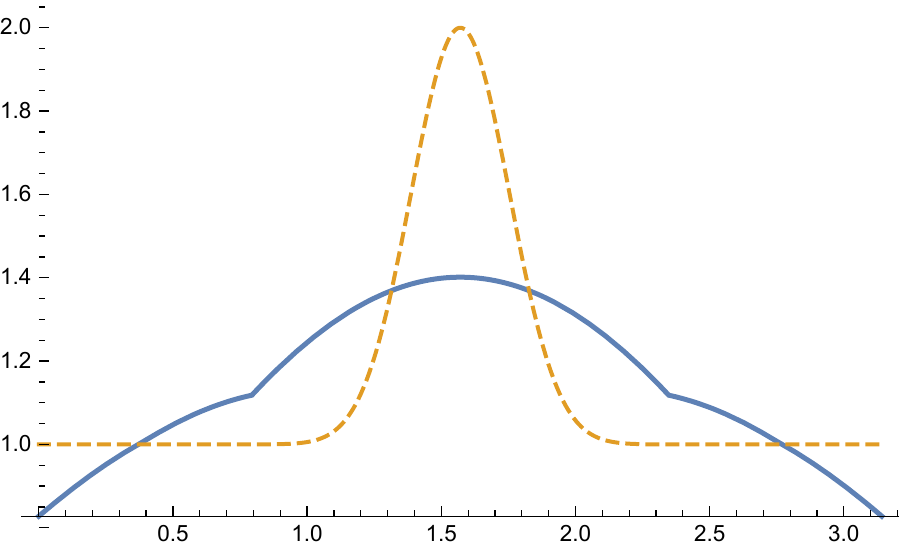}
\caption{$k=5,M=20,\epsilon =0.53$}
\end{subfigure}
\begin{subfigure}{0.49\textwidth}
\includegraphics[width=0.45\textwidth]{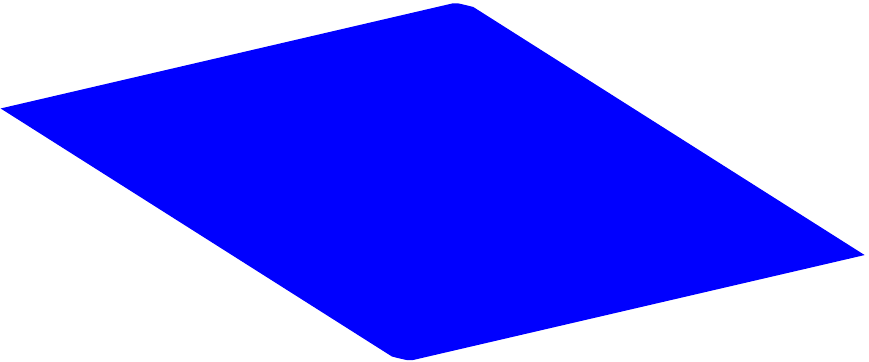}
\includegraphics[width=0.45\textwidth]{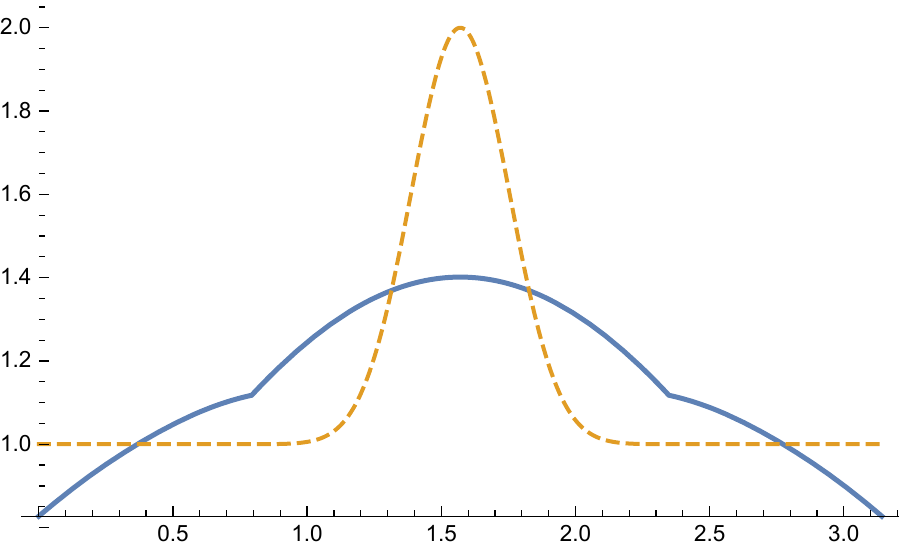}
\caption{$k=20,M=20,\epsilon =0.53$}
\end{subfigure}
\caption{Reconstructed bodies for $f(\theta)=1+\sin(\theta)^{30}$}\label{figDirac}
\end{center}
\end{figure*}
\end{example}

\section{Relationship to Onsager}\label{secOnsagerComparison}

Let $\mm$ be a compact convex body in $\mathbb{R}^n$. Then we define the Onsager free energy, ${\mf(\cdot,\mm):L^1(\mbox{SO}(n))\to\mathbb{R}}$ by 
\begin{equation}\begin{split}
\mf(\rho,\mm)=&\int_{\mbox{SO}(n)}\rho(R)\ln \rho(R)\,dR+\frac{1}{2}\int_{\mbox{SO}(n)}\int_{\mbox{SO}(n)}\rho(R)\rho(S)Ex(RS^T,\mm)\,dR\,dS.
\end{split}\end{equation}
We minimise the free energy subject to the density constraint, that ${\int_{\mbox{SO}(n)}\rho(R)\,dR=\rho_0}$. The value of $\rho_0$ is a parameter that mediates the competition between order and disorder in our system, which we will take to be fixed for this section. 

The aim of this section is to demonstrate that, at the level of the Onsager model, the metric of Hausdorff convergence is appropriate in the sense that if $\mm_i\to\mm$, then the functionals $\mf(\cdot,\mm_i)$ $\Gamma$-converge to $\mf(\cdot,\mm)$ (see \Cref{defGammaConvergence}). 

While the previous results were for bodies in $\mathbb{R}^2$, the following discussion is readily performed in arbitrary dimension. However before proceeding, we will need an important result on the regularity of $Ex(\cdot,\mm)$. We proved in \Cref{propW1inftyBound} that if $\mm$ is a zonoid in $\mathbb{R}^2$, then $Ex(\cdot,\mm)$ is Lipschitz. We now prove the result for general convex bodies in $\mathbb{R}^n$.

\begin{proposition}\label{propExLip}
Let $\mm$ be a convex body in $\mathbb{R}^n$. Then $Ex(\cdot,\mm):\mbox{SO}(n)\to\mathbb{R}$ is Lipschitz. Furthermore, if $\mathcal{K}$ is a set of convex bodies, precompact in Hausdorff metric, then $\sup\limits_{\mm \in \mathcal{K}}\mbox{Lip}\, Ex(\cdot,\mm)<+\infty$.
\end{proposition} 
\begin{proof}
Within this proof we will have to consider bodies $\mm_R=\mm-R\mm$. From the compactness of $\mbox{SO}(n)$ and continuity of operations with Hausdorff metric, ${\{\mm_R: R\in\mbox{SO}(n)\}}$ is compact with respect to Hausdorff metric for fixed $\mm$, and the set ${\{\mm_R:\mm\in\mathcal{K}, \, R\in\mbox{SO}(n)\}}$ is precompact with respect to Hausdorff metric. In particular, this means that the intrinsic volumes of such $\mm_R$ are uniformly bounded. Thus, by the Steiner formula, we have that if the Hausdorff metric $d_H({\mm}_R,{\mm}_{R'})<\delta$, then we can find constants so that 
\begin{equation}\begin{split}
&|{\mm}_R|-C_1\delta\\
\leq &|{\mm}_R|-\delta P(\delta,{\mm}_{R'})\\
\leq &|{\mm}_{R'}|\\
\leq &|{\mm}_R|+\delta P(\delta,{\mm}_{R})\\
\leq &|{\mm}_R|+C_1\delta.
\end{split}\end{equation}
In this case, $P(\cdot,\mm_R)$ is an $n-1$ degree polynomial in $\delta$, with $P(0,\mm)=\per(\mm)$, and generally the coefficients are given by intrinsic volumes of the body. In particular, $C_1$ is a constant depending smoothly only on the intrinsic volumes of $\mm$, and more so the constant is uniformly bounded for $\mm\in\mathcal{K}$, as intrinsic volumes are continuous with the Hausdorff metric. Then \begin{equation}\begin{split}
&|Ex(R,\mm)-Ex(R',\mm)|\\
=& \left||{\mm}_R|-|{\mm}_{R'}|\right|\\
\leq & C_1d_H({\mm}_R,{\mm}_{R'})\\
=& C_1\sup\limits_{\eta\in\mathbb{S}^{n-1}}|h(\eta,{\mm}_R)-h(\eta,{\mm}_{R'})|\\
=&C_1\sup\limits_{\eta\in\mathbb{S}^{n-1}}\left|h(\eta,{\mm})+h(-R^T\eta,{\mm})-h(\eta,{\mm})-h(-R'^T\eta,{\mm})\right|\\
=&C_1\sup\limits_{\eta\in\mathbb{S}^{n-1}}\left|h(-R^T\eta,{\mm})-h(-R'^T\eta,{\mm})\right|\\
\leq&   C_1\mbox{Lip}(h(\cdot,\mm))||R-R'||.
\end{split}\end{equation}
We remark that the Lipschitz constant of the support functions must be uniformly bounded if they correspond to a precompact set of convex bodies. This is because precompact sets in Hausdorff metric are necessarily uniformly bounded, and the Lipschitz constant of the support function can be bounded by the diameters as  
\begin{equation}\begin{split}
&h(p,\mm)-h(q,\mm)\\
=& \sup\limits_{x \in \mm} p\cdot x - \sup\limits_{x\in\mm}q\cdot x\\
=& \sup\limits_{x \in \mm}( (p-q)\cdot x+q\cdot x) - \sup\limits_{x\in\mm}q\cdot x\\
\leq & \sup\limits_{x \in \mm} (p-q)\cdot x +\sup\limits_{x\in\mm}q\cdot x- \sup\limits_{x\in\mm}q\cdot x\\
=& \sup\limits_{x \in \mm} (p-q)\cdot x\leq  |p-q|\sup\limits_{x\in\mm}|x|. 
\end{split}\end{equation}
\end{proof}

\begin{corollary}
If $\mm_i\to\mm$ in Hausdorff metric, then $Ex(\cdot,\mm_i)\to Ex(\cdot,\mm)$ weak-* in $W^{1,\infty}$.
\end{corollary}
\begin{proof}
It suffices to prove pointwise convergence, as we have uniform control on the Lipschitz constant. However, this is immediate as $\mm\mapsto\mm-R\mm$ is a continuous function in Hausdorff metric for fixed $R \in \mbox{SO}(n)$. 
\end{proof}

We recall the definition of $\Gamma$-convergence \protect\cite{braides2002gamma}.

\begin{definition}\label{defGammaConvergence}
Let $X$ be a topological space, and $F_i:X\to\mathbb{R}\cup\{+\infty\}$ be functionals. We say that $F_i$ $\Gamma$-converge to $F:X\to\mathbb{R}\cup\{+\infty\}$, with respect to the topology on $X$, if 
\begin{enumerate}
\item (Liminf inequality) For every $x\in X$ and every sequence $x_i\to x$, $\liminf\limits_{i\to\infty}F_i(x_i)\geq F(x)$.
\item (Limsup inequality) For every $x\in X$, there exists a sequence $x_i\to x$ with $\limsup\limits_{i\to\infty}F_i(x_i)=F(x)$.
\end{enumerate}
\end{definition}

Furthermore, we have the fundamental theorem of $\Gamma$-convergence.

\begin{theorem}
Let $F_i$ $\Gamma$-converge to $F$. Furthermore, assume the sequence $F_i$ is equicoercive, so that if $x_i$ is a sequence in $X$ and $F_i(x_i)$ is uniformly bounded, then there exists a subsequence $x_{i_j}$ and some $x\in X$ with $x_{i_j}\to x$. Then we have that minimisers of $F$ exist, $\lim\limits_{i\to\infty} (\inf F_i)\to \min F$, and if $x_i$ is a sequence with $\lim\limits_{i\to\infty} F_i(x_i)-\inf F_i\to 0$, then there exists a subsequence $x_{i_j}$ and minimiser $x$ of $F$ with $x_{i_j}\to x$. 
\end{theorem}

Loosely speaking, this theorem states that $\Gamma$-convergence is a ``natural" mode of convergence for minimisation problems, in the sense that minimisers of $F_i$ will converge to minimisers of $F$ if $F_i$ $\Gamma$-converges to $F$. 

\begin{theorem}\label{theoremOnsagerConverge}
Let $\mm_i\to \mm$ in Hausdorff metric. Then $\mf(\cdot,\mm_i)$ $\Gamma$-converges to $\mf(\cdot,\mm)$ with respect to weak-$L^1$ convergence. Furthermore $\mf(\cdot,\mm_i)$ are equicoercive with respect to the weak-$L^1$ topology.
\end{theorem}
\begin{proof}
First we must show that if $\mf(\rho_i,{\mm}_i)$ is bounded, then $\rho_i$ admits a weakly converging subsequence. Furthermore, we must show that $\min\mf(\cdot,{\mm_i})$ is bounded. For the latter, we simply test the uniform distribution, as 
\begin{equation}\begin{split}
&\min\mf(\cdot,{\mm_i})\\
\leq & \mf\left(\frac{\rho_0}{|\mbox{SO}(n)|},{\mm}_i\right)\\
\leq &-\rho_0\ln\frac{\rho_0}{|\mbox{SO}(n)|}+\frac{\rho_0^2}{2} ||Ex(\cdot,{\mm}_i)||_\infty,
\end{split}\end{equation} where the Hausdorff convergence ensures the latter is uniformly bounded. Then, we see that as $Ex$ is always non-negative, that 
\begin{equation}\begin{split}
\int_{\mbox{SO}(n)}\rho(R)\ln\rho(R)\,dR\leq \mf(\rho,{\mm}_i).
\end{split}\end{equation}
This implies that the Shannon entropy is bounded if $\mf(\rho_i,{\mm}_i)$ is bounded, which implies a weak $L^1$-converging subsequence by the theorem of de la Vall\'ee Poussin \protect\cite[Theorem 22]{meyer1966probability}.

Next we must show that if $\rho_i\overset{L^1}{\rightharpoonup}\rho$, then $\liminf\limits_{\i \to\infty}\mf(\rho_i,{\mm}_i)\geq \mf(\rho,\mm)$. For brevity, define the operators $T_{\mm'}:L^1(\mbox{SO}(n))\to L^\infty (\mbox{SO}(n))$ by 
\begin{equation}\begin{split}
T_{\mm'}\rho(R)=\int_{\mbox{SO}(n)}\rho(S)Ex(RS^T,{\mm}')\,dS.\end{split}\end{equation} As $Ex$ is Lipschitz, these are compact operators from $L^1$ to $L^\infty$, which follows by the Arzel\'a-Ascoli theorem.

Now we note the energy has two components. The bilinear term in $\mf(\cdot,{\mm_i})$ is of the form $\langle \rho,T_{\mm_i}\rho\rangle$, where $T_{\mm_i}$ are compact operators converging to $T_{\mm}$. This means that if $\rho_i\rightharpoonup\rho$ in $L^1$, $T_{\mm_i}\rho_i\to T_{\mm}\rho$ in $L^\infty$, so $\langle \rho,T_{\mm_i}\rho_i\rangle\to \langle \rho,T_{\mm}\rho\rangle$. The Shannon entropy term is convex, hence lower semicontinuous with respect to weak-$L^1$ convergence, and independent of $i$, so that 
\begin{equation}\begin{split}
&\liminf\limits_{i\to\infty} \int_{\mbox{SO}(n)}\rho_i(R)\ln\rho_i(R)\,dR\\
\geq &\liminf\limits_{i\to\infty} \int_{\mbox{SO}(n)}\rho(R)\ln\rho(R)\,dR.
\end{split}\end{equation}
Combining these, we have the liminf inequality,
\begin{equation}\begin{split}
\liminf\limits_{\i \to\infty}\mf(\rho_i,{\mm}_i)\geq \mf(\rho,\mm).
\end{split}\end{equation}

Finally, we show the so-called limsup inequality. For the sake of this work, it suffices to show that $\mf(\rho,{\mm_i})\to\mf(\rho,\mm)$ for all $\rho$ with finite Shannon entropy. This is however straightforward, as $\mf(\rho,\mm_i)-\mf(\rho,\mm)=\langle \rho,T_{{\mm}_i}\rho\rangle-\langle \rho,T_{{\mm}}\rho\rangle$, which converges to zero by the same argument as given in proving the liminf inequality. 
\end{proof}

\begin{remark}
The equilibrium equations for the Onsager model are given by 
\begin{equation}\begin{split}
\ln\rho(R)=\lambda -\int_{\mbox{SO}(n)}Ex(RS^T)\rho(S)\,dS,
\end{split}\end{equation}
where $\lambda$ is a Lagrange multiplier corresponding to the constraint $\int_{\mbox{SO}(n)}\rho(R)\,dR=\rho_0$ \protect\cite{constantin2010high}. As $Ex$ is Lipschitz, bootstrapping type arguments could provide much stronger modes of convergence for minimisers of $\mf(\cdot,\mm_i)$ than weak $L^1$, though these arguments are tedious and will be omitted for future work.
\end{remark}

\section{Conclusions}
Within this work we have derived expressions for the excluded area function for 2D convex bodies, in terms of the Fourier coefficients of the body's support function. Using this formula, it is shown that for well behaved shapes there are are typically uncountably many such bodies with the same excluded volume functions, with some exceptions, notably a disk. The formula also permitted the derivation and analysis of an algorithm that can construct a body from a candidate excluded volume function $f$, whose excluded volume function is the best $L^2$-approximation over all excluded volume functions in an appropriate space. Finally, a comparison with Onsager demonstrates that Hausdorff convergence, used throughout the work, is indeed a ``natural" mode of convergence for convex bodies with regards to expected equilibria. 

\section{Acknowledgments}
The author would like to express gratitude to Peter Palffy-Muhoray, Epifanio Virga,and Mark Wilkinson and Xiaoyu Zheng, all of whom have provided useful discussions on the subject matter in this work. 
\section{Bibliography}
\bibliography{Zonopaper}

\end{document}